\documentclass[11pt]{article}
\usepackage[title,titletoc]{appendix}
\usepackage[margin=1in]{geometry}
\usepackage[T1]{fontenc}
\usepackage{lmodern}
\usepackage{microtype}
\usepackage{amsmath,amssymb,amsthm,mathtools}
\usepackage{booktabs,tabularx,array}
\usepackage{enumitem}
\usepackage[round,authoryear]{natbib}
\usepackage{graphicx}
\usepackage[hidelinks]{hyperref}

\allowdisplaybreaks
\numberwithin{equation}{section}
\setlist[enumerate]{itemsep=0.25em,topsep=0.5em}
\setlist[itemize]{itemsep=0.2em,topsep=0.4em}

\newtheorem{theorem}{Theorem}[section]
\newtheorem{proposition}[theorem]{Proposition}
\newtheorem{lemma}[theorem]{Lemma}
\newtheorem{corollary}[theorem]{Corollary}
\theoremstyle{definition}
\newtheorem{definition}[theorem]{Definition}
\newtheorem{assumption}[theorem]{Assumption}
\newtheorem{example}[theorem]{Example}
\theoremstyle{remark}
\newtheorem{remark}[theorem]{Remark}

\newcommand{\PP}{\mathbb P}
\newcommand{\EE}{\mathbb E}
\newcommand{\TV}{\operatorname{TV}}
\newcommand{\cL}{\mathcal L}
\newcommand{\1}{\mathbf 1}
\newcommand{\wbar}{\bar w}
\newcommand{\Del}{\Delta_0}
\newcommand{\Lb}{\bar\ell}
\newcommand{\rb}{\bar r}
\newcommand{\sigH}{\sigma^H}

\title{Washed Out by the Crowd?\\[0.25em]
\large Accountability under Sequential Review}
\author{%
  {\large Siming Ye}\footnote{%
    Department of Economics, \emph{Georgetown University}. 
    Email: \texttt{sy677@georgetown.edu}. \\
    I owe my greatest debt to Daniel Rappoport, who has guided this project from its earliest stages and shaped its development at every turn. I also thank my advisor, Christopher Chambers, for his continued encouragement and support. I further thank Kevin He, Roger Lagunoff, Axel Anderson and Garance Genicot for their helpful comments. All errors are my own.
  }%
}
\date{\textbf{Preliminary Draft} \\[1em] \today}
\begin{document}
\maketitle

\begin{abstract}
When an early information producer is judged only after others have reviewed and revised the work, the same review that sharpens the final decision can blur the question of who deserves the credit. This paper asks how an organization can still reward careful early work once a chain of later reviewers has acted on it. In the model, an analyst's hidden effort makes an initial report more likely to be right; a sequence of reviewers then reacts to it; and the organization can pay only on the record this process leaves behind. The main result splits the value of any such record into two parts: how much effort improved the first report, and how well the final record still indicates whether the first report was accurate. When later reviewers can overturn a flawed report, review improves decisions but washes the analyst's effort out of the final outcome; therefore, rewarding the final outcome stops working, while rewarding agreement between the first and last word has incentive value instead. When the first report is simply copied by downstream reviewers, the reverse holds. Which reward is better comes down to one thing: how likely review is to repair an initial mistake.
\end{abstract}


\section{Introduction}\label{sec:introduction}

Many organizational decisions begin with an initial assessment and close only after later people have reviewed it. A junior analyst screens a project before an investment committee decides; a compliance officer flags a transaction before legal or managerial review; a claims adjuster recommends a disposition before supervisors settle the file. The effort behind the initial assessment is costly and difficult to observe. The record left by the subsequent review is easier to verify.

This paper asks how an organization can use that later record to reward careful early work. The difficulty is that the review process itself responds to the report under scrutiny: reviewers see earlier recommendations, receive fresh information, and act on their own incentives to match the state. A strong review can correct a flawed initial report; a weak one may simply repeat it. The same scrutiny that improves the final decision can therefore blur what the record says about the analyst who filed the initial report.

The model embeds the hidden-action problem in a binary review chain. The state is $\theta\in\{-1,+1\}$; an analyst chooses costly, unobserved effort and then produces a public report $a_0\in\{-1,+1\}$, with higher effort making the report more likely to match the state. Downstream reviewers act in sequence, and the principal can pay the analyst only on a record that is contractible when the process settles. The central question is whether that record still separates histories that began with a correct first report from histories that began with a mistaken one.

Two incomplete records organize the comparison. Under outcome evaluation, the principal verifies the state and the final action, and the natural favorable event is terminal correctness. Under process evaluation, the principal archives the initial and final recommendations, and the natural favorable event is first--last agreement. Agreement carries information only when review could have produced disagreement: when reviewers mechanically defer, it is a rubber stamp; when they tend to overturn flawed reports, it is evidence that the first report survived scrutiny.

The main result is an accounting identity. Effort reaches any downstream record only through the probability that the initial report is correct, and conditional on that event the review process unfolds identically after high effort and after a secret deviation. The incentive content of a record therefore factors into two terms: the primitive effect of effort on first-report accuracy, and the record's ability to distinguish the correct-start branch from the mistaken-start branch.

The memory-1 benchmark supplies the dynamics behind the title. Each reviewer observes only her predecessor's action and one bounded private signal. When the initial report is contestable, the public belief climbs toward the level at which a single private signal can just no longer overturn the inherited action. Along that path, terminal correctness grows insensitive to the analyst's hidden effort---the crowd washes it out---while first--last agreement gains incentive content, because correct reports are preserved and mistaken reports are repaired. When the initial report is uncontestable, and reviewers copy it, the ranking flips: terminal correctness retains the full effort imprint, and agreement is degenerate.

The remainder of the paper proceeds as follows. Section~\ref{sec:literature} reviews the related literature. Before specifying the full learning model, Section~\ref{sec:simple-example} provides a reduced-form example to illustrate the core mechanisms of preservation and repair. We formally define the baseline environment in Section~\ref{sec:environment}, which allows us to establish the main record-strength and branch-decomposition results in Section~\ref{sec:records}. Sections~\ref{sec:memory1} through~\ref{sec:numerical} solve and numerically illustrate the memory-1 benchmark, focusing on the separation between the on-path belief path and the hidden-deviation law. Finally, Sections~\ref{sec:beyond} to~\ref{sec:asymmetry} consider extensions of the baseline setup—including non-Markovian dynamics, retained traces, and asymmetric priors—before Section~\ref{sec:discussion} concludes.

\section{Related literature}\label{sec:literature}

The paper connects moral hazard with social learning. The downstream review social learning process belongs to the literature on Bayesian observational learning and cascades, including \citet{Banerjee1992}, \citet{BikhchandaniHirshleiferWelch1992}, \citet{SmithSorensen2000}, \citet{AcemogluDahlehLobelOzdaglar2011}, and the finite-memory environment of \citet{DrakopoulosOzdaglarTsitsiklis2013}. This branch of literature asks when a sequence of actions aggregates information about the underlying state. A more recent study, \citet{KartikLeeLiuRappoport2024}, identifies a general condition, \textit{excludability}, under which preferences and information jointly support learning on observational networks, and shows that the standard unbounded-beliefs benchmark can be too strong in multi-state environments. 

Deviating from Bayesian learning environments, \citet{JadenChen2026} show that once agents are uncertain about others' signal structures, informational cascades can arise from an asymmetry in worst-case beliefs under max-min preferences, rather than from the standard comparison between accumulated public evidence and the strength of a private signal. \citet{SmithSorensenTian2021} study herding and optimal experimentation from a normative standpoint, and one of their findings is that efficient incentives reward agents who are later imitated, a reward for being agreed with, in a different guise, that reappears in the analysis below.

These papers ask how well a sequence of actions learns or uses information about the state. This paper asks what the same downstream process preserves of an earlier, hidden choice of effort: whether the record it leaves behind still carries a contractible trace of that choice.

The contracting angle of this project roots in the informativeness principle in moral hazard \citep{Holmstrom1979} and related work on performance measurement \citep{Baker1992,ChaigneauEdmansGottlieb2014}. Bounded payments matter because a fixed wage cap can turn weak statistical attribution into outright infeasibility, as \citet{JewittKadanSwinkels2008} show. Because the analyst is herself an information producer, the model also relates to delegated
information acquisition \citep{ChadeKovrijnykh2016, Carroll2019, ClarkReggiani2021, WhitmeyerZhang2023, Wu2026} and to the contractible or partially contractible evaluation of expert information \citep{PapireddygariWaggoner2022, Xia2025}. In our context, by contrast, downstream reviewers transform the original report before any record is settled and compensation is paid, so evaluation reaches the analyst only through the transformed record.

A related strand of studies asks what information about an agent's action should be made public in order to discipline behavior. \citet{Prat2005} shows that observing an agent's action directly, rather than only its consequence, can backfire by inducing conformity. A related tension appears below between rewarding a final outcome and rewarding agreement with an initial report, though for a different reason: there, the principal is constrained by whatever record a review process happens to leave behind, rather than choosing a disclosure policy that itself changes the analyst's incentive to report honestly. \citet{HalacLipnowskiRappoport2021} instead show how private, rank-contingent contracts can implement a desired action profile as the unique equilibrium.

More broadly, this paper examines how individuals interpret others' interpretations and how that information is passed down. In a continuous-action social learning experiment, \citet{AngrisaniGuarinoJehielKitagawa2021} find that subjects weight their own signal approximately correctly but place less-than-Bayesian, roughly constant weight on the information in predecessors' actions -- consistent with relative overconfidence, or mistrust of predecessors' ability to read their own signals, rather than with redundancy neglect, which would instead overweight early movers. A complementary departure runs the other way: \citet{DasarathaHe2020} show that agents who naively neglect correlation among their neighbors' observations can overweight early movers, so that some network structures make the first report disproportionately influential. This paper keeps Bayesian downstream reviewers as its clean benchmark, but nothing in the argument depends on that choice: whatever rule reviewers actually follow, Bayesian or not, that rule alone pins down how often a correct initial report survives review and how often an incorrect one is corrected, and it is those two probabilities that carry all of the analyst's incentives.

Finally, the paper relates to work on organizational information processing and communication in hierarchies \citep{SahStiglitz1986, BoltonDewatripont1994, MarschakReichelstein1998,
Garicano2000, AryaGloverMittendorf2006}, which studies how organizations aggregate and route information. This paper asks a complementary question: holding an information-routing architecture fixed, how does it change what a later record reveals about an earlier, hidden choice of effort? The chain structure used below, in which each person observes only a message left by an immediate predecessor rather than the full history, appears outside social learning as well. \citet{AnderliniGerardiLagunoff2010} study a similarly dynastic, message-based architecture in which social memory and evidence shape conflict rather than the aggregation of information about a state.

\section{A simple example: preservation, repair, and attribution}\label{sec:simple-example}

Consider a binary task. The state of the world is $\theta\in\{-1,+1\}$, and reports use the same two labels. A report is correct when it equals the state. An analyst first chooses effort $e\in\{H,L\}$ and then produces the initial report $a_0$. Effort is not observed by the reviewers or by the principal. Its direct effect is only to make the first report more accurate:
\begin{equation}\label{eq:simple-alpha}
  \alpha_e:=\PP(a_0=\theta\mid e),
  \qquad
  \alpha_H>\alpha_L.
\end{equation}
Write
\begin{equation}\label{eq:simple-delta}
  \Del:=\alpha_H-\alpha_L>0
\end{equation}
for the primitive gain in first-report accuracy. High effort therefore shifts probability mass from mistaken starts to correct starts.

After the initial report, downstream reviewers conduct the review. A reviewer may observe the current report, later evidence, or earlier review activity, depending on the procedure. The reviewer does not observe the analyst's effort, and while reviewing she does not know the state. Her strategy is the rule that maps what she observes into the next reviewed report. Those strategies may use beliefs about whether the inherited report is correct, but the beliefs and the resulting choices are part of the same review procedure.

The important point for incentives is that the same review procedure is used after either effort choice. A low-effort analyst does not announce that she chose low effort, and reviewers do not see it. Low effort therefore shifts probability toward mistaken starts while leaving the treatment of any given review history unchanged. All probabilities below are computed under this fixed procedure.

Let $a_T$ denote the final reviewed report. For the two endpoint records considered here, the whole review procedure can be summarized by two conditional probabilities:
\begin{equation}\label{eq:simple-SR}
  S:=\PP(a_T=\theta\mid a_0=\theta),
  \qquad
  R:=\PP(a_T=\theta\mid a_0\ne\theta).
\end{equation}
The number $S$ is preservation. Starting from a correct first report, review ends with a correct final report with probability $S$. The number $R$ is repair. Starting from a mistaken first report, review corrects the mistake by the end with probability $R$.

Suppose first that the settlement record verifies the state and the final reviewed report, but not the analyst's initial report. The principal can then reward terminal correctness, $a_T=\theta$. Under effort $e$, the review starts from a correct report with probability $\alpha_e$ and from a mistaken report with probability $1-\alpha_e$. Hence
\begin{equation}\label{eq:simple-out-prob}
  \PP(a_T=\theta\mid e)=\alpha_e S+(1-\alpha_e)R.
\end{equation}
The high-minus-low probability gap for terminal correctness is therefore
\begin{equation}\label{eq:simple-out-gap}
  \PP(a_T=\theta\mid H)-\PP(a_T=\theta\mid L)=\Del(S-R).
\end{equation}
Terminal correctness rewards effort only when a correct first report is more likely than a mistaken first report to lead to a correct final report. If review makes the final report equally accurate after both kinds of starts, terminal correctness measures the quality of review but no longer attributes much to the analyst's effort.

Now suppose instead that the settlement record keeps the analyst's initial report and the final reviewed report, but not the state. The principal can then reward first--last agreement, $a_T=a_0$. On a correct start, agreement means that the correct first report survived review, so it occurs with probability $S$. On a mistaken start, agreement means that the mistake survived review, so it occurs with probability $1-R$. Thus
\begin{equation}\label{eq:simple-agree-prob}
  \PP(a_T=a_0\mid e)=\alpha_e S+(1-\alpha_e)(1-R),
\end{equation}
and the high-minus-low probability gap for first--last agreement is
\begin{equation}\label{eq:simple-agree-gap}
  \PP(a_T=a_0\mid H)-\PP(a_T=a_0\mid L)=\Del(S+R-1).
\end{equation}
Agreement rewards effort only when review tends to preserve correct starts and reject mistaken starts. If every initial report is rubber-stamped, then $S=1$ and $R=0$; agreement occurs under both effort levels and has no incentive content.

Both records point in the intended direction when
\begin{equation}\label{eq:simple-positive-orientation}
  S>R,
  \qquad
  S+R>1.
\end{equation}
The first inequality says that terminal correctness is more likely after a correct first report than after a mistaken one. The second says that first--last agreement is more likely after a correct first report than after a mistaken one, since the agreement probabilities on the two starts are $S$ and $1-R$.

\begin{proposition}[Endpoint incentive comparison]\label{prop:simple-threshold}
Suppose \eqref{eq:simple-positive-orientation} holds. Then first--last agreement has a weakly larger high-minus-low probability gap than terminal correctness if and only if
\begin{equation}\label{eq:simple-half}
  R\ge \frac12.
\end{equation}
Consequently, if high effort has incremental cost $\kappa>0$ and the contract pays a bonus only when the chosen record succeeds, the minimum bonus required under first--last agreement is weakly smaller than the minimum bonus required under terminal correctness if and only if $R\ge1/2$.
\end{proposition}

\begin{proof}
Under \eqref{eq:simple-positive-orientation}, both probability gaps are positive. Equations \eqref{eq:simple-out-gap} and \eqref{eq:simple-agree-gap} show that first--last agreement has the larger gap exactly when
\[
  \Del(S+R-1)\ge \Del(S-R).
\]
Because $\Del>0$, this inequality is equivalent to $S+R-1\ge S-R$, or $R\ge1/2$. For the bonus statement, a bonus $b$ paid on a record raises the analyst's expected payoff from high effort relative to low effort by $b$ times that record's probability gap. The minimum implementing bonus is therefore $\kappa$ divided by the relevant positive gap. The record with the larger gap requires the smaller bonus.
\end{proof}

The threshold has a simple interpretation. Preservation $S$ helps both records in the same way: after a correct first report, both terminal correctness and first--last agreement benefit when that report survives review. The records differ after a mistaken first report. Repair helps terminal correctness, but failed repair creates agreement for the wrong reason. Once repair is more likely than failed repair, agreement turns into evidence that the first report survived a review process that usually catches mistakes.

\section{Model and contractible records}\label{sec:environment}

\subsection{State, effort, and the initial report}\label{subsec:initial}

The state is binary,
\[
  \theta\in\{-1,+1\},
  \qquad
  \PP(\theta=+1)=\PP(\theta=-1)=\frac12.
\]
An initial analyst, indexed by $0$, chooses hidden effort
\[
  e\in\{H,L\}.
\]
The analyst is risk neutral. Let
\[
  \kappa:=k(H)-k(L)>0
\]
be the incremental cost of high effort.

The investigation produces a binary public report
\[
  a_0\in\{-1,+1\}.
\]
The report is the recorded output of the investigation technology, fixed once the investigation closes; the analyst has no separate reporting stage in which to revise it. This isolates moral hazard in information production from the problem of truthful reporting, which Section~\ref{subsec:scope} discusses. The report's accuracy is symmetric:
\begin{equation}\label{eq:initial-tech}
  \PP(a_0=\theta\mid e)=\alpha_e,
  \qquad
  1>\alpha_H>\alpha_L>\frac12.
\end{equation}
Write
\begin{equation}\label{eq:Delta}
  \Del:=\alpha_H-\alpha_L>0.
\end{equation}
The parameter $\Del$ is the primitive statistical imprint of effort on the initial public report.

\begin{remark}[A primitive signal interpretation]\label{rem:primitive-signal}
One can obtain \eqref{eq:initial-tech} from a private investigation signal $s_0\in\{-1,+1\}$ satisfying $\PP(s_0=\theta\mid e)=\alpha_e$, together with an automatic reporting rule $a_0=s_0$. Nothing in the analysis requires the analyst to make a downstream state-matching decision. The analyst produces information; later agents use it.
\end{remark}

\subsection{Downstream review and the fixed-continuation comparison}\label{subsec:downstream}

There are downstream reviewers $t=1,\ldots,T$. Reviewer $t$ chooses
\[
  a_t\in\{-1,+1\}
\]
and has state-matching utility
\[
  u_t(a_t,\theta)=\1\{a_t=\theta\}.
\]
The information available to reviewer $t$ depends on the review architecture. In the memory-1 benchmark, reviewer $t$ observes only $(a_{t-1},x_t)$, where $x_t$ is a fresh private signal. Section~\ref{sec:beyond} allows arbitrary public histories.

Effort is not observed by downstream reviewers. The intended equilibrium has high effort, so reviewers interpret the initial report using the accuracy $\alpha_H$. Let
\[
  \sigH
\]
denote the resulting Bayesian continuation strategy profile. For each actual effort level $e$, define
\begin{equation}\label{eq:Pe}
  \PP^e(\cdot):=\PP(\cdot\mid e,\sigH),
  \qquad
  \EE^e[\cdot]:=\EE[\cdot\mid e,\sigH].
\end{equation}
Thus $\PP^H$ is the on-path law. Under $\PP^L$, the analyst secretly chooses low effort while every downstream reviewer continues to use the high-effort continuation rule.

The comparison holds the continuation fixed because the deviation is secret: a low-effort analyst changes the distribution of $a_0$ and, through it, the distribution of every later history, while each reviewer keeps applying the thresholds implied by $\sigH$. Recomputing the continuation from the accuracy $\alpha_L$ would instead describe a chain in which low effort is publicly anticipated, and that object never enters the analyst's incentive constraint.

\subsection{Downstream exogeneity}\label{subsec:exogeneity}

Let $Z_T$ collect all downstream primitive randomness and public variables other than $(\theta,a_0)$, including private signals, randomization devices, and later actions. The maintained restriction is that effort has no direct effect on these objects once the state and the initial report are fixed.

\begin{assumption}[Downstream exogeneity]\label{ass:exogeneity}
There exists a probability kernel $K$ such that, for every effort $e$,
\begin{equation}\label{eq:kernel-factorization}
  \PP^e(a_0=a,Z_T\in B\mid\theta)
  =q_e(a\mid\theta)K(B\mid\theta,a),
\end{equation}
where $q_e$ is the initial-report technology implied by \eqref{eq:initial-tech}, and $K$ does not depend on $e$. The downstream strategy profile embedded in $K$ is fixed at $\sigH$.
\end{assumption}

Assumption~\ref{ass:exogeneity} allows effort to change the distribution of the first report but rules out a second effort-dependent trace after $(\theta,a_0)$ is fixed. If reviewers observe raw notes, elapsed investigation time, or a confidence score whose conditional distribution depends on effort, that variable must be included in the initial public state. The two-branch decomposition below can then be generalized, but it cannot be applied while silently omitting that trace.

\subsection{Contractible records}\label{subsec:records-environment}

Let
\[
  h_T:=(a_0,a_1,\ldots,a_T)
\]
be the public action history. A contractible record is a measurable function of the objects verifiable at settlement. The paper compares two benchmark technologies.

\begin{definition}[Outcome evaluation]\label{def:outcome}
At settlement the principal observes $(\theta,a_T)$, but $a_0$ is unavailable or noncontractible. The canonical favorable event is
\[
  E_T^\theta:=\{a_T=\theta\}.
\]
\end{definition}

\begin{definition}[Process evaluation]\label{def:process}
At settlement the principal observes $(a_0,a_T)$, but $\theta$ is unavailable or noncontractible. The canonical favorable event is
\[
  E_T^A:=\{a_0=a_T\}.
\]
\end{definition}

Table~\ref{tab:technologies} separates the information environments. Under hidden truth in the symmetric model, a lone terminal action carries no incentive information (Corollary~\ref{cor:single-action}); the settlement record must retain the state or a link back to the initial report. If both $\theta$ and $a_0$ are contractible, direct audit dominates delayed endpoint evaluation for the narrow purpose of inducing the initial effort.

\begin{table}[t]
\centering
\caption{Settlement information and the natural record}
\label{tab:technologies}
\begin{tabularx}{0.92\textwidth}{@{}>{\raggedright\arraybackslash}p{0.20\textwidth}>{\raggedright\arraybackslash}X>{\raggedright\arraybackslash}X@{}}
\toprule
 & Initial report $a_0$ unavailable & Initial report $a_0$ contractible \\
\midrule
State $\theta$ unavailable
& A terminal action alone is marginally uninformative under sign symmetry.
& Process evaluation uses $(a_0,a_T)$ and rewards agreement or disagreement. \\
\addlinespace
State $\theta$ contractible
& Outcome evaluation uses $(\theta,a_T)$ and rewards terminal correctness or error.
& Direct audit uses $(\theta,a_0)$; downstream review is unnecessary for initial-effort incentives. \\
\bottomrule
\end{tabularx}
\end{table}

The distinction can reflect technological or institutional constraints. Historical reports may be overwritten, paraphrased, or legally unusable in compensation; conversely, the truth may never be verified while version histories keep the initial and final recommendations easy to compare. The analysis treats these as different contractibility regimes: each regime specifies which record is available for compensation at settlement.

\paragraph{Timing.}
Before effort is chosen, the principal commits to an evaluation technology and a wage schedule. Nature then draws the state, the analyst chooses effort, and the investigation technology produces $a_0$. Reviewers act sequentially under the continuation rule induced by the equilibrium belief about effort. At date $T$, the contractible record is revealed and the committed wage is paid. Neither the analyst's wage nor the settlement record enters reviewers' state-matching payoffs.

\subsection{Wages and implementation}\label{subsec:implementation}

A wage is bounded by limited liability and an upper cap,
\begin{equation}\label{eq:wage-bounds}
  0\le w\le\wbar.
\end{equation}
The analyst receives no intrinsic payoff from the state or downstream actions. High effort is implemented if
\begin{equation}\label{eq:IC}
  \EE^H[w]-\EE^L[w]\ge\kappa.
\end{equation}
Participation is suppressed, or equivalently is handled by a base salary outside the incremental bonus problem. The wage $w$ and cap $\wbar$ below should therefore be read as incremental performance pay. The analysis concerns the additional cost of satisfying \eqref{eq:IC}.

For a contractible record $G$, define
\begin{equation}\label{eq:cost}
  C^*(G;\kappa)
  :=
  \inf_{0\le w\le\wbar}\EE^H[w(G)]
\end{equation}
subject to \eqref{eq:IC}, with $C^*(G;\kappa)=+\infty$ if the constraint is infeasible. When the state is part of the record, it is included in $G$.

\begin{table}[t]
\centering
\caption{Main notation}
\label{tab:notation}
\begin{tabularx}{\textwidth}{@{}lX@{}}
\toprule
Symbol & Meaning \\
\midrule
$\theta$ & Binary state. \\
$e\in\{H,L\}$ & Initial analyst's hidden effort. \\
$\alpha_e$ & Accuracy of the initial report under effort $e$. \\
$\Del=\alpha_H-\alpha_L$ & Primitive effort effect on initial-report accuracy. \\
$\kappa$ & Incremental cost of high effort. \\
$a_t$ & Initial report or downstream public action at date $t$. \\
$\sigH$ & Downstream continuation rule induced by the high-effort interpretation of the initial report. \\
$\hat r_t$ & On-path belief path: probability under $\PP^H$ that $a_t$ is correct; generates reviewer thresholds. \\
$r_t(e)$ & Actual correctness probability under effort $e$, evaluated under the fixed continuation $\sigH$. \\
$S_T$ & Preservation: terminal correctness after a correct initial report. \\
$R_T$ & Repair: terminal correctness after a mistaken initial report. \\
$D(G)$ & Total-variation distance between the high- and low-effort laws of record $G$. \\
\bottomrule
\end{tabularx}
\end{table}

Three of these objects carry the analysis. The on-path belief path $(\hat r_t)$ generates the thresholds reviewers apply; the deviation accuracy $r_t(L)$ evolves under those fixed thresholds; and the record distance $D(G)$ caps the incentive power available from bounded wages.

\section{Information content of downstream records}\label{sec:records}

This section gives a record distance a direct incentive meaning: with bounded wages, that distance is exactly the largest incentive gap the principal can generate from the record. It then traces every record's value to a single channel, the separation between the branch where the initial report was correct and the branch where it was mistaken.


For any contractible record $G$, define
\begin{equation}\label{eq:D}
  D(G):=\TV\!\left(\cL^H(G),\cL^L(G)\right),
\end{equation}
where
\[
  \TV(P,Q):=\sup_B|P(B)-Q(B)|.
\]
The definition applies whether or not $G$ includes the state.

\begin{proposition}[Maximum incentive gap]\label{prop:maxgap}
For every record $G$,
\begin{equation}\label{eq:maxgap}
  \sup_{0\le w\le\wbar}
  \left\{\EE^H[w(G)]-\EE^L[w(G)]\right\}
  =\wbar D(G).
\end{equation}
Consequently, some bounded wage on $G$ implements high effort if and only if
\begin{equation}\label{eq:feasibility}
  \wbar D(G)\ge\kappa.
\end{equation}
\end{proposition}

The result gives $D(G)$ a direct economic interpretation: it is the maximum incentive gap per unit of the wage cap. The gap-maximizing wage pays on record values that are relatively more likely after high effort, the cells with $\PP^H(G=g)>\PP^L(G=g)$ when $G$ is finite and the region where the high-effort density exceeds the low-effort density otherwise. When a record is said below to lose incentive power, this contracting object is what shrinks.

For finite records, the minimum-cost contract has a closed form: the principal fills the cells that deliver the most incentive per expected high-effort dollar.

\begin{proposition}[Minimum-cost contract on a finite record]\label{prop:finite-contract}
Suppose $G$ takes values $g_1,\ldots,g_n$. Let
\[
  p_i^e:=\PP^e(G=g_i),
  \qquad
  \delta_i:=p_i^H-p_i^L,
\]
and let $I_+:=\{i:\delta_i>0\}$. For $i\in I_+$ define the high-effort cost per unit of incentive
\[
  \rho_i:=\frac{p_i^H}{\delta_i}.
\]
Order the indices in $I_+$ so that $\rho_1\le\cdots\le\rho_m$. If \eqref{eq:feasibility} holds, let $k$ be the smallest index satisfying
\[
  \wbar\sum_{i=1}^k\delta_i\ge\kappa.
\]
There is a minimum-cost contract with
\begin{equation}\label{eq:finite-optimal-wage}
  w_i=
  \begin{cases}
  \wbar, & i<k,\\[0.2em]
  \displaystyle
  \frac{\kappa-\wbar\sum_{j<k}\delta_j}{\delta_k}, & i=k,\\[0.8em]
  0, & \text{otherwise}.
  \end{cases}
\end{equation}
If several cells have the same cutoff ratio, payments can be redistributed among those tied cells without changing cost.
\end{proposition}

The proposition is a fractional-knapsack characterization. A cell with $\delta_i>0$ is evidence in favor of high effort. The principal first pays on the cells that generate the most incentive per expected high-effort dollar. Cells with $\delta_i\le0$ are never rewarded in a minimum-cost contract under limited liability.

\begin{corollary}[Binary record]\label{cor:binary-cost}
Suppose the record reduces to an event $E$ with
\[
  p_H:=\PP^H(E)>p_L:=\PP^L(E).
\]
The minimum-cost contract pays zero on $E^c$ and
\begin{equation}\label{eq:binary-bonus}
  b^*:=\frac{\kappa}{p_H-p_L}
\end{equation}
on $E$. It is feasible if and only if $b^*\le\wbar$, and its high-effort expected cost is
\begin{equation}\label{eq:binary-cost}
  \frac{\kappa p_H}{p_H-p_L}.
\end{equation}
\end{corollary}


Let
\[
  C:=\{a_0=\theta\},
  \qquad
  M:=\{a_0\ne\theta\}.
\]
For a record $G$, define the conditional laws
\begin{equation}\label{eq:branch-laws}
  \nu_G^C:=\cL(G\mid C,\sigH),
  \qquad
  \nu_G^M:=\cL(G\mid M,\sigH).
\end{equation}
The record may include the state. Under the symmetric initial technology, the conditional distribution of $(\theta,a_0)$ within either branch does not depend on effort. On $C$, the pairs $(+1,+1)$ and $(-1,-1)$ each have probability one half; on $M$, the pairs $(+1,-1)$ and $(-1,+1)$ each have probability one half.

\begin{theorem}[Initial-report decomposition]\label{thm:decomposition}
Under Assumption~\ref{ass:exogeneity}, for every record $G$ and effort $e$,
\begin{equation}\label{eq:mixture}
  \cL^e(G)=\alpha_e\nu_G^C+(1-\alpha_e)\nu_G^M.
\end{equation}
Equivalently, for every event $B$ in the record space,
\begin{equation}\label{eq:event-gap}
  \PP^H(G\in B)-\PP^L(G\in B)
  =\Del\bigl[\nu_G^C(B)-\nu_G^M(B)\bigr].
\end{equation}
Therefore
\begin{equation}\label{eq:D-decomp}
  D(G)=\Del\,\TV(\nu_G^C,\nu_G^M).
\end{equation}
In particular,
\begin{equation}\label{eq:D-upper}
  D(G)\le\Del
\end{equation}
for every downstream record.
\end{theorem}

Equation \eqref{eq:D-decomp} is the main accounting identity of the paper. It separates the analyst's information-production technology, summarized by $\Del$, from the monitoring content of the record, summarized by the distance between the correct-report and incorrect-report branches. The second term can be zero even when effort strongly affects $a_0$, and it can equal one only when the record reveals the initial branch perfectly.

The upper bound has two interpretations. First, downstream review cannot create statistical information about effort that was absent from the initial report. Second, a rich downstream trace can nevertheless make the initial effort contractible when the state itself is not. The review process transforms the original signal into a new record; the transformation may preserve or destroy branch separation.


\begin{proposition}[Coarsening cannot improve a record]\label{prop:data-processing}
If $G'$ is a measurable function of $G$, then
\begin{equation}\label{eq:data-processing}
  D(G')\le D(G).
\end{equation}
More specifically,
\[
  \TV(\nu_{G'}^C,\nu_{G'}^M)
  \le
  \TV(\nu_G^C,\nu_G^M).
\]
\end{proposition}

Thus the full public history is an upper benchmark for every statistic computed from it. Let
\[
  H_T:=(a_0,a_1,\ldots,a_T).
\]
Then
\begin{equation}\label{eq:full-history-benchmark}
  D(H_T)
  =\Del\,\TV\!\left(
      \cL(H_T\mid a_0=\theta),
      \cL(H_T\mid a_0\ne\theta)
    \right),
\end{equation}
and every endpoint record measurable from $H_T$ has weakly smaller distance. First--last agreement is therefore a transparent benchmark within a larger family: when the full history is archived, Proposition~\ref{prop:finite-contract} describes the payments a principal would spread across richer patterns.

\begin{corollary}[Direct audit]\label{cor:direct-audit}
If $(\theta,a_0)$ is contractible, then
\begin{equation}\label{eq:direct-D}
  D(\theta,a_0)=\Del.
\end{equation}
Adding any downstream history cannot increase the distance:
\[
  D(\theta,a_0,a_1,\ldots,a_T)=\Del.
\]
When feasible, the minimum-cost direct-audit contract rewards $\{a_0=\theta\}$ with bonus $\kappa/\Del$ and has expected high-effort cost
\begin{equation}\label{eq:direct-cost}
  C^{\mathrm{direct}}=\frac{\kappa\alpha_H}{\Del}.
\end{equation}
\end{corollary}

Direct audit reaches the upper bound because the branches $C$ and $M$ occupy disjoint cells once both the state and the initial report are observed. The formula also fixes the benchmark for the two settlement regimes: each replaces one coordinate of the direct-audit pair with the terminal action and pays on what remains.

\begin{corollary}[A single public action under hidden truth]\label{cor:single-action}
Suppose the downstream environment and strategy profile are sign symmetric: flipping the state and every action leaves induced probabilities unchanged. Then, for every date $t$ and effort $e$,
\begin{equation}\label{eq:single-uniform}
  \PP^e(a_t=+1)=\frac12.
\end{equation}
Consequently,
\[
  D(a_t)=0
\]
for every $t$, including $t=0$.
\end{corollary}

A hidden-truth principal therefore cannot infer effort from the sign of one action. The useful process information is relational: whether two actions agree, whether a later action reverses an earlier one, or more generally which history was generated.


How a record's value moves with the review horizon depends on how the downstream architecture treats the two branches. A process that eventually reaches the truth on both branches drives the branch laws of the terminal action together, so terminal correctness loses its grip on initial effort; the same process separates the branch laws of agreement, because correct reports survive while mistaken ones are reversed. Mechanical copying does the opposite: it ties the terminal action to the initial report and makes agreement certain. The next sections derive both effects in a model whose full path can be solved.

\section{Memory-1 social learning}\label{sec:memory1}

This section focuses specifically on a memory-1 review chain, in which each reviewer observes only her predecessor's action and one private signal. This section defines the on-path belief path and solves its dynamics; the path generates the thresholds reviewers use. Section~\ref{sec:hidden} then introduces hidden deviations and measures what those fixed thresholds transmit of the analyst's effort.


Reviewer $t$ observes $(a_{t-1},x_t)$. Conditional on $\theta$, the private signals $(x_t)_{t\ge1}$ are independent and identically distributed and are independent of the initial report. Let $f_+$ and $f_-$ denote their conditional densities and define the private log-likelihood ratio
\begin{equation}\label{eq:private-llr}
  \ell_t:=\log\frac{f_+(x_t)}{f_-(x_t)}.
\end{equation}

\begin{assumption}[Symmetric bounded experiment]\label{ass:signals}
There is $\Lb\in(0,\infty)$ such that $\ell_t\in[-\Lb,\Lb]$ almost surely. Under $\theta=+1$, $\ell_t$ has a continuous density $h$ that is strictly positive on $(-\Lb,\Lb)$, with cumulative distribution function $H$. The binary experiment is sign symmetric:
\begin{equation}\label{eq:signal-symmetry}
  \cL(\ell_t\mid\theta=-1)
  =\cL(-\ell_t\mid\theta=+1).
\end{equation}
There are no atoms at the support endpoints.
\end{assumption}

The LLR property and sign symmetry are distinct restrictions. The first relates the density under the two states at the same LLR realization; the second identifies the density under one state at opposite realizations.

\begin{lemma}[LLR density identity]\label{lem:llr-identity}
Under Assumption~\ref{ass:signals},
\begin{equation}\label{eq:llr-identity}
  h(-z)=e^{-z}h(z),
  \qquad z\in(-\Lb,\Lb).
\end{equation}
\end{lemma}

The largest private LLR magnitude is $\Lb$. Define
\begin{equation}\label{eq:rbar}
  \rb:=\frac{e^{\Lb}}{1+e^{\Lb}}.
\end{equation}
An inherited public action with correctness probability at least $\rb$ carries an LLR magnitude of at least $\Lb$, so no private signal can strictly overturn it and the next reviewer copies. We call $\rb$ the \emph{copying frontier}: the confidence level at which bounded private information loses the power to reverse the inherited action.

\subsection{The on-path belief path}\label{subsec:public-state}

Reviewers interpret the initial report through the equilibrium belief that effort was high, and the continuation rule $\sigH$ encodes that interpretation. Define
\begin{equation}\label{eq:rhat}
  \hat r_t:=\PP^H(a_t=\theta),
\end{equation}
the probability under the on-path law that the public action at date $t$ matches the state. The initial value is
\[
  \hat r_0=\alpha_H,
\]
because reviewers read the analyst's report as the output of a high-effort investigation. We call the sequence $(\hat r_t)_{t\ge0}$ the \emph{on-path belief path}.

The path plays one role in everything that follows: it generates the thresholds reviewers apply to their private signals. Lemma~\ref{lem:action-symmetry} shows that $\hat r_t$ equals the posterior probability that $a_t$ is correct for an observer who sees only $a_t$, so each public action carries a log-likelihood ratio determined by the current coordinate of the path. A secret low-effort deviation leaves every threshold in place and changes the distribution of histories generated under those thresholds; Section~\ref{sec:hidden} tracks that second object through the actual correctness probabilities $r_t(e)$, with $r_t(H)=\hat r_t$ on path.

One scalar suffices for the public state because the prior is symmetric, the signal experiment is sign symmetric, and reviewer $t+1$ observes only $a_t$: at every date, the public history relevant for the next action compresses into the identity of the current action and the common correctness level $\hat r_t$.

\begin{lemma}[Public LLR carried by an action]\label{lem:action-symmetry}
For every $t$,
\begin{equation}\label{eq:posterior-correctness}
  \PP^H(\theta=a_t\mid a_t)=\hat r_t.
\end{equation}
Consequently, if
\begin{equation}\label{eq:s-def}
  s_t:=\log\frac{\hat r_t}{1-\hat r_t},
\end{equation}
then observing $a_t$ contributes public log-likelihood ratio $a_t s_t$ for the event $\theta=+1$.
\end{lemma}

Ignoring zero-probability ties, reviewer $t+1$ therefore chooses
\begin{equation}\label{eq:action-rule}
  a_{t+1}=+1
  \quad\Longleftrightarrow\quad
  a_t s_t+\ell_{t+1}\ge0.
\end{equation}
Equation \eqref{eq:action-rule} is the complete belief-to-action link in the benchmark: the inherited action enters through the threshold $s_t$, and the threshold is generated by the on-path belief path. Because effort is hidden, reviewer $t+1$ applies the same rule after either effort choice; a deviation reaches her only through the joint distribution of $(a_t,\ell_{t+1})$.


For $r\in(1/2,\rb)$, let
\[
  s(r):=\log\frac{r}{1-r}
\]
and define
\begin{equation}\label{eq:F}
  F(r)
  :=
  r\bigl[1-H(-s(r))\bigr]
  +(1-r)\bigl[1-H(s(r))\bigr].
\end{equation}
The first term is the probability that a correct predecessor remains correct. The second is the probability that an incorrect predecessor is overturned.

\begin{proposition}[On-path learning map]\label{prop:learning-map}
Whenever $\hat r_t\in(1/2,\rb)$,
\begin{equation}\label{eq:rhat-recursion}
  \hat r_{t+1}=F(\hat r_t).
\end{equation}
Moreover,
\begin{equation}\label{eq:Fprime}
  F'(r)=H(s(r))-H(-s(r))\in(0,1),
\end{equation}
and $F'$ is strictly increasing on $(1/2,\rb)$.
\end{proposition}

The derivative formula \eqref{eq:Fprime} relies on the LLR identity of Lemma~\ref{lem:llr-identity}: differentiating the thresholds produces density terms that cancel exactly because the signal is its own likelihood ratio.

\begin{corollary}[Two memory-1 regimes]\label{cor:regimes}
The on-path belief path satisfies one of two alternatives.

\begin{enumerate}[label=(\roman*)]
\item \textbf{Contestable initial report.} If $\alpha_H<\rb$, then
\begin{equation}\label{eq:rhat-frontier}
  \hat r_t\uparrow\rb.
\end{equation}
The sequence stays strictly below $\rb$ at every finite date.

\item \textbf{Immediate copying.} If $\alpha_H\ge\rb$, then
\begin{equation}\label{eq:immediate-copy}
  a_t=a_0
\end{equation}
almost surely for every $t\ge1$.
\end{enumerate}
\end{corollary}

In the first regime, on-path accuracy converges to $\rb<1$: public confidence settles at the point where a single bounded signal just loses the power to overturn the inherited action, and learning of the state stops there. The washout results below concern a second, distinct limit, the sensitivity of downstream records to the analyst's hidden effort along this path.

\subsection{The influence coefficient}\label{subsec:influence}

For an on-path state $\hat r_t<\rb$, define
\begin{equation}\label{eq:lambda}
  \lambda_t
  :=H(s_t)-H(-s_t)
  =F'(\hat r_t).
\end{equation}
The coefficient has two equivalent interpretations. First, it is the derivative of next-period public accuracy with respect to current public accuracy. Second, holding the public threshold $s_t$ fixed, it is the change in the next reviewer's correctness probability when the predecessor changes from incorrect to correct.

Under a common private-signal coupling and state $\theta=+1$, changing the predecessor action from $-1$ to $+1$ changes the next action exactly when
\[
  \ell_{t+1}\in[-s_t,s_t),
\]
which has probability $\lambda_t$. This is the event on which the predecessor's correctness stays pivotal for the next action.

The probability that reviewer $t+1$ copies her predecessor is a separate object. On path it equals
\begin{equation}\label{eq:copy-probability}
  q_t^{\mathrm{copy}}
  =
  \hat r_t\bigl[1-H(-s_t)\bigr]
  +(1-\hat r_t)H(s_t).
\end{equation}
The two expressions generally differ: a reviewer may copy because the predecessor is right or because the predecessor and signal mislead her together, while $\lambda_t$ counts only the realizations at which the inherited action is pivotal.

The distinction matters for the product formula below: the effort imprint travels only through pivotal signal realizations, and the attenuation product $\Pi_T$ multiplies their probabilities along the chain.

\section{Hidden effort and endpoint evaluation}\label{sec:hidden}

Fix the high-effort continuation rule $\sigH$ (downstream agents believe that high effort is exerted). This section asks how the laws of contractible records change when the analyst shirks, secretly deviating to low effort. The relevant downstream thresholds remain those generated by the on-path belief path $(\hat r_t)$.

\subsection{Propagation of a hidden deviation}\label{subsec:hidden-recursion}

Define the actual correctness probability under effort $e$ by
\begin{equation}\label{eq:actual-r}
  r_t(e):=\PP^e(a_t=\theta),
\end{equation}
and the high-versus-low gap by
\begin{equation}\label{eq:d-def}
  d_t:=r_t(H)-r_t(L).
\end{equation}
On path, $r_t(H)=\hat r_t$. Under a deviation, $r_t(L)$ starts from the lower accuracy $\alpha_L$ and evolves under the same thresholds $(s_t)$, because reviewers continue to use $\sigH$.

Suppose $\alpha_H<\rb$, and define
\begin{equation}\label{eq:b-def}
  b_t:=1-H(s_t),
\end{equation}
with $s_t$ generated by the on-path belief path.

\begin{proposition}[Hidden-deviation recursion]\label{prop:hidden-recursion}
Suppose $\alpha_H<\rb$. For each $e\in\{H,L\}$,
\begin{equation}\label{eq:actual-recursion}
  r_{t+1}(e)=b_t+\lambda_t r_t(e),
  \qquad
  r_0(e)=\alpha_e.
\end{equation}
Consequently,
\begin{equation}\label{eq:d-product}
  d_T=\Del\Pi_T,
  \qquad
  \Pi_T:=\prod_{t=0}^{T-1}\lambda_t,
  \qquad
  \Pi_0:=1.
\end{equation}
\end{proposition}

The recursion is affine because thresholds respond to the on-path belief while actual effort enters only through the predecessor's correctness. The intercept $b_t$ is the probability that an incorrect predecessor gets corrected; the slope $\lambda_t$ is the incremental effect of predecessor correctness. Differencing the two effort paths removes the intercept, so the effort gap travels through the product of influence coefficients alone.

\begin{corollary}[Washout of terminal accuracy]\label{cor:washout}
If $\alpha_H<\rb$, then
\begin{equation}\label{eq:Pi-zero}
  \Pi_T\to0
  \qquad\text{and}\qquad
  d_T\to0.
\end{equation}
Moreover, $\Pi_T$ and $d_T$ are strictly decreasing in $T$.
\end{corollary}

Along the frontier path, reviewers' on-path accuracy stays bounded by $\rb<1$ while the gap $d_T$ closes: under the fixed high-effort continuation rule, terminal accuracy ends up nearly the same whichever effort the analyst actually chose. The state remains imperfectly learned; the effort imprint is what vanishes.

\subsection{Preservation and repair branches}\label{subsec:branches}

The preservation and repair probabilities of Section~\ref{sec:simple-example} now acquire a horizon index. Define
\begin{equation}\label{eq:SR}
  S_T:=\PP(a_T=\theta\mid a_0=\theta,\sigH),
  \qquad
  R_T:=\PP(a_T=\theta\mid a_0\ne\theta,\sigH).
\end{equation}
The number $S_T$ is the probability that review preserves a correct initial report through date $T$; the number $R_T$ is the probability that review repairs a mistaken initial report by date $T$.

\begin{proposition}[Branch probabilities in memory-1]\label{prop:SR-memory1}
If $\alpha_H<\rb$, then
\begin{align}
  S_T&=\hat r_T+(1-\alpha_H)\Pi_T,\label{eq:S}\\
  R_T&=\hat r_T-\alpha_H\Pi_T.\label{eq:R}
\end{align}
Hence
\begin{align}
  S_T-R_T&=\Pi_T,\label{eq:SminusR}\\
  S_T+R_T-1
  &=2\hat r_T-1+(1-2\alpha_H)\Pi_T.\label{eq:SplusR}
\end{align}
\end{proposition}

The difference $S_T-R_T$ measures how much terminal correctness still remembers whether the first report was correct. The sum $S_T+R_T-1$ measures how strongly review preserves correct starts while rejecting mistaken ones; this combination governs the informativeness of agreement.

\begin{theorem}[Opposite monotonicity of the two branch indices]\label{thm:monotonicity}
In the memory-1 frontier regime $\alpha_H<\rb$:
\begin{enumerate}[label=(\roman*)]
\item $S_T-R_T=\Pi_T$ is strictly decreasing from one to zero;
\item $R_T$ is strictly increasing from zero to $\rb$;
\item $S_T+R_T-1$ is strictly increasing from zero to $2\rb-1$.
\end{enumerate}
\end{theorem}

Review progressively erases the branch difference in terminal \emph{accuracy} and progressively widens the branch difference in \emph{agreement with the first action}: the same learning process weakens one performance measure while it strengthens the other, date by date.

\subsection{Endpoint record laws}\label{subsec:endpoint-laws}

Under sign symmetry, the four cells of $(\theta,a_T)$ reduce to correctness versus error, and the four cells of $(a_0,a_T)$ reduce to agreement versus disagreement. The branch probabilities therefore yield closed-form distances.

\begin{proposition}[Endpoint distinguishability]\label{prop:endpoint-D}
For every sign-symmetric downstream architecture,
\begin{align}
  D(\theta,a_T)
  &=\Del\,|S_T-R_T|,\label{eq:D-out-general}\\
  D(a_0,a_T)
  &=\Del\,|S_T+R_T-1|.\label{eq:D-agree-general}
\end{align}
In the memory-1 frontier regime,
\begin{align}
  D(\theta,a_T)&=\Del\Pi_T,\label{eq:D-out-memory}\\
  D(a_0,a_T)&=\Del\bigl[S_T+R_T-1\bigr].\label{eq:D-agree-memory}
\end{align}
\end{proposition}

Combining Proposition~\ref{prop:endpoint-D} with Theorem~\ref{thm:monotonicity} gives a strict horizon comparison.

\begin{corollary}[Outcome weakens while agreement strengthens]\label{cor:opposite-records}
If $\alpha_H<\rb$, then $D(\theta,a_T)$ is strictly decreasing in $T$, while $D(a_0,a_T)$ is strictly increasing in $T$. Their limits are
\begin{equation}\label{eq:record-limits-preview}
  \lim_{T\to\infty}D(\theta,a_T)=0,
  \qquad
  \lim_{T\to\infty}D(a_0,a_T)=\Del(2\rb-1).
\end{equation}
\end{corollary}

\subsection{Bonuses, wage caps, and expected costs}\label{subsec:cost-comparison}

For $T\ge1$ in the frontier regime, both relevant orientations are positive:
\[
  S_T>R_T,
  \qquad
  S_T+R_T>1.
\]
Thus outcome evaluation rewards terminal correctness and process evaluation rewards agreement.

\begin{proposition}[Minimum-cost endpoint contracts]\label{prop:endpoint-costs}
The minimum success bonuses are
\begin{align}
  b_T^\theta
  &=\frac{\kappa}{\Del(S_T-R_T)}
   =\frac{\kappa}{\Del\Pi_T},\label{eq:b-out}\\
  b_T^A
  &=\frac{\kappa}{\Del(S_T+R_T-1)}.\label{eq:b-agree}
\end{align}
Each contract is feasible exactly when its required bonus does not exceed $\wbar$.

When feasible, the minimum high-effort expected costs are
\begin{align}
  C_T^\theta
  &=
  \frac{\kappa\,[\alpha_H S_T+(1-\alpha_H)R_T]}
       {\Del(S_T-R_T)}
   =\frac{\kappa\hat r_T}{\Del\Pi_T},\label{eq:C-out}\\
  C_T^A
  &=
  \frac{\kappa\,[\alpha_H S_T+(1-\alpha_H)(1-R_T)]}
       {\Del(S_T+R_T-1)}.\label{eq:C-agree}
\end{align}
These contracts are optimal among all bounded contracts on their respective endpoint records.
\end{proposition}

The outcome bonus rises strictly with the horizon because $\Pi_T$ falls. The agreement bonus falls strictly because $S_T+R_T-1$ rises. A longer review chain therefore imposes opposite wage-cap requirements on the two technologies.

\begin{theorem}[The repair threshold]\label{thm:repair-threshold}
At every finite horizon $T\ge1$ in the memory-1 frontier regime, the following statements are equivalent:
\begin{align*}
  D(a_0,a_T)&\ge D(\theta,a_T),\\
  b_T^A&\le b_T^\theta,\\
  R_T&\ge\frac12.
\end{align*}
If both endpoint contracts are feasible, then
\begin{equation}\label{eq:cost-ranking}
  C_T^A\le C_T^\theta
  \quad\Longleftrightarrow\quad
  R_T\ge\frac12.
\end{equation}
More precisely,
\begin{equation}\label{eq:cost-difference}
  C_T^A-C_T^\theta
  =-
  \frac{\kappa S_T(2R_T-1)}
       {\Del(S_T-R_T)(S_T+R_T-1)}.
\end{equation}
\end{theorem}

The threshold has a direct interpretation. If bad initial reports are usually left uncorrected, agreement is easy to obtain for the wrong reason and terminal correctness is the stronger signal of initial effort. Once a bad initial report is more likely than not to be repaired, agreement becomes demanding: it tends to occur on the correct-report branch and fail on the incorrect-report branch.

\begin{corollary}[A unique switch in the preferred endpoint record]\label{cor:unique-switch}
If $\alpha_H<\rb$, define
\begin{equation}\label{eq:switch-horizon}
  T^\times
  :=\min\left\{T\ge1:R_T\ge\frac12\right\}.
\end{equation}
Then $T^\times$ is finite. Outcome evaluation is strictly stronger for $T<T^\times$, process evaluation is weakly stronger at $T=T^\times$, and process evaluation is strictly stronger for every $T>T^\times$ unless equality happens exactly at the crossing date.
\end{corollary}

The unique switch follows from strict monotonicity of $R_T$, and it occurs along the review horizon itself: for a fixed organization and fixed signal technology, the preferred endpoint record changes as the number of review stages grows.

\subsection{Long-horizon duality}\label{subsec:duality}

\begin{theorem}[Endpoint duality in memory-1]\label{thm:duality}
Fix $\Del>0$, $\kappa>0$, and $\wbar<\infty$.

\begin{enumerate}[label=(\roman*)]
\item \textbf{Contestable initial report.} If $\alpha_H<\rb$, then
\[
  \hat r_T\uparrow\rb,
  \qquad
  \Pi_T\to0,
  \qquad
  S_T,R_T\to\rb.
\]
Consequently,
\begin{align}
  D(\theta,a_T)&\to0,\label{eq:out-dies}\\
  D(a_0,a_T)&\to\Del(2\rb-1)
  =\Del\tanh(\Lb/2)>0.\label{eq:agree-survives}
\end{align}
Outcome evaluation is eventually infeasible under every fixed $(\wbar,\kappa)$. Process evaluation is feasible for all sufficiently large $T$ if and only if
\begin{equation}\label{eq:process-asym-feasible}
  \wbar\Del(2\rb-1)>\kappa.
\end{equation}
If $\wbar\Del(2\rb-1)\le\kappa$, it is infeasible at every finite horizon.

\item \textbf{Immediate copying.} If $\alpha_H\ge\rb$, then $a_T=a_0$ almost surely for every $T\ge1$, and
\begin{align}
  D(\theta,a_T)&=\Del,\label{eq:out-copy}\\
  D(a_0,a_T)&=0.\label{eq:agree-copy}
\end{align}
Outcome evaluation preserves the full initial effort effect, while process evaluation is exactly degenerate.
\end{enumerate}
\end{theorem}

In the frontier regime, the limiting process contract has bonus
\begin{equation}\label{eq:process-limit-bonus}
  b_\infty^A
  =\frac{\kappa}{\Del(2\rb-1)}
\end{equation}
and expected high-effort cost
\begin{align}
  C_\infty^A
  &=
  \frac{\kappa\,[\alpha_H\rb+(1-\alpha_H)(1-\rb)]}
       {\Del(2\rb-1)}\label{eq:process-limit-cost-a}\\
  &=
  \frac{\kappa\alpha_H}{\Del}
  +
  \frac{\kappa(1-\rb)}{\Del(2\rb-1)}.\label{eq:process-limit-cost-b}
\end{align}
The first term in \eqref{eq:process-limit-cost-b} is the direct-audit cost from \eqref{eq:direct-cost}. The second is the price of using an imperfect downstream verifier when the truth is not contractible. Even though agreement retains positive incentive power, it cannot match the cost of observing initial-report correctness itself.

Under the additional frontier regularity stated in Appendix~\ref{app:rates}, the contrast is quantitative: outcome distinguishability is $\Theta(T^{-2})$, so its required bonus and the associated uncapped wage bill grow as $\Theta(T^2)$ before any fixed wage cap makes implementation infeasible. Process distinguishability approaches its positive limit at rate $O(T^{-1})$; when its limiting bonus is below the cap, its bonus and expected cost converge to finite limits.

\section{Numerical Example}\label{sec:numerical}

In this section, a specified example shows the speed of the two opposing record dynamics and verifies that the signal family is a genuine LLR experiment. The analytical results themselves hold for any distribution satisfying Assumption~\ref{ass:signals}.

\subsection{A bounded LLR family}

Fix $\Lb>0$. Under $\theta=+1$, let the LLR density be
\begin{equation}\label{eq:numerical-density}
  h(z)=
  \begin{cases}
  c e^z, & -\Lb<z<0,\\
  c, & 0\le z<\Lb,
  \end{cases}
  \qquad
  c:=\frac{1}{\Lb+1-e^{-\Lb}}.
\end{equation}
The corresponding CDF is
\begin{equation}\label{eq:numerical-cdf}
  H(z)=
  \begin{cases}
  0, & z\le-\Lb,\\
  c(e^z-e^{-\Lb}), & -\Lb<z<0,\\
  c(1-e^{-\Lb}+z), & 0\le z<\Lb,\\
  1, & z\ge\Lb.
  \end{cases}
\end{equation}
For $z>0$, $h(-z)=ce^{-z}=e^{-z}h(z)$, so the density satisfies the valid-LLR identity. Defining the state-$-1$ distribution by sign reversal produces a symmetric binary experiment. Moreover, $h(z)/h(-z)=e^z$, so the observation $z$ is indeed its own log-likelihood ratio rather than an arbitrary threshold index. The density is continuous and strictly positive on the open support, and it is smooth in a neighborhood of the upper frontier.

This example uses
\begin{equation}\label{eq:numerical-parameters}
  \Lb=1.2,
  \qquad
  \alpha_H=0.65,
  \qquad
  \alpha_L=0.55.
\end{equation}
Then
\[
  \Del=0.10,
  \qquad
  \rb=\frac{e^{1.2}}{1+e^{1.2}}\approx0.7685,
\]
so the initial report is in the contestable regime.

\subsection{Belief, repair, and record strength}

Figure~\ref{fig:branch-paths} plots the on-path accuracy $\hat r_T$, the correct-report branch $S_T$, and the repair probability $R_T$. The on-path state rises toward $\rb$. The two branch probabilities approach the same limit, which is why terminal correctness loses its dependence on the initial effort. At the same time, $R_T$ crosses one half between horizons two and three. The preferred endpoint record therefore switches at $T^\times=3$.

\begin{figure}[htbp!]
\centering
\includegraphics[width=0.78\textwidth]{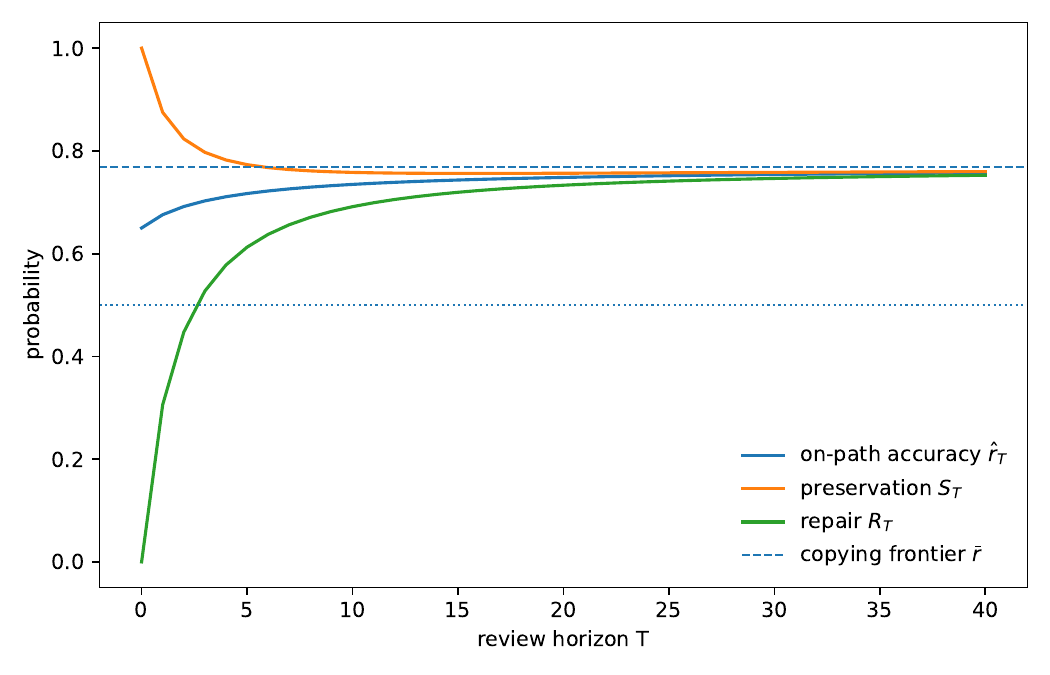}
\caption{On-path accuracy and branch probabilities. Parameters are given in \eqref{eq:numerical-parameters}. The dotted horizontal line is one half; the dashed line is the copying frontier $\rb$.}
\label{fig:branch-paths}
\end{figure}

Figure~\ref{fig:record-strength} shows the corresponding total-variation distances. For normalization, the figure includes $T=0$, where the outcome record coincides with the direct audit $(\theta,a_0)$ and its distinguishability equals the full initial-report distance $\Del$. Thereafter outcome distinguishability falls monotonically. Agreement distinguishability begins at zero and rises toward
\[
  \Del(2\rb-1)\approx0.0537.
\]
The curves cross at the same horizon at which $R_T$ crosses one half, as Theorem~\ref{thm:repair-threshold} requires.

\begin{figure}[htbp!]
\centering
\includegraphics[width=0.78\textwidth]{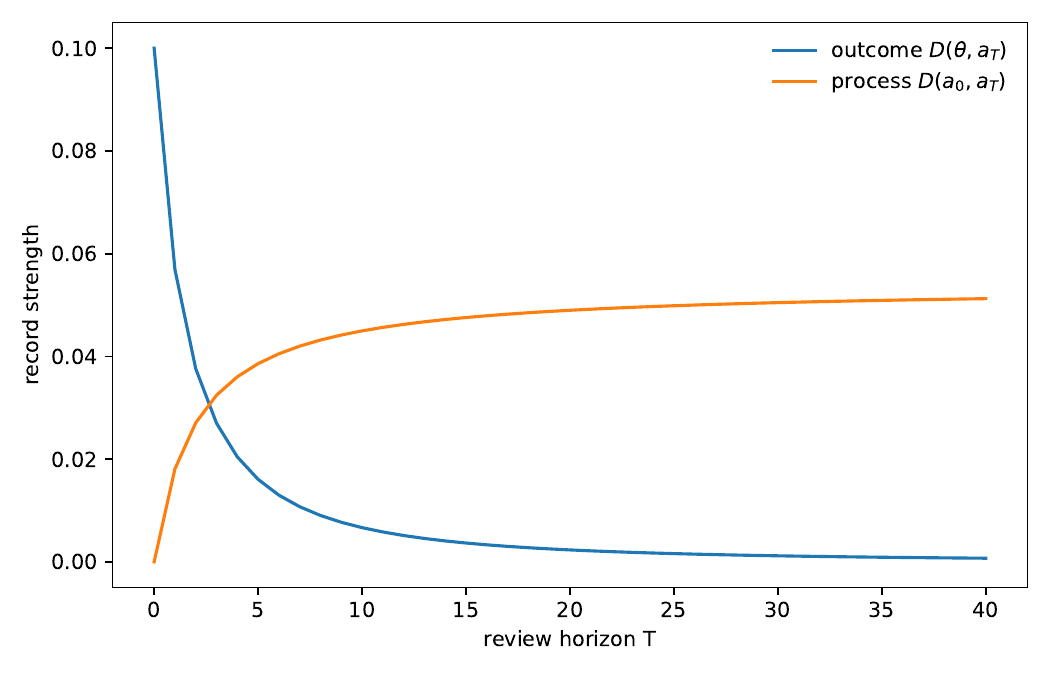}
\caption{Incentive power of the two endpoint records. The outcome record weakens with the review horizon, while the process record strengthens. The $T=0$ outcome point is the direct-audit normalization $(\theta,a_0)$.}
\label{fig:record-strength}
\end{figure}

Table~\ref{tab:numerical} reports selected values. The attenuation product falls quickly at first and then polynomially near the frontier. The repair probability, by contrast, passes one half at a short horizon and continues toward $\rb$.

\begin{table}[t]
\centering
\caption{Selected numerical values}
\label{tab:numerical}
\begin{tabular}{@{}rrrrrr@{}}
\toprule
$T$ & $\hat r_T$ & $R_T$ & $\Pi_T$ & $D(\theta,a_T)$ & $D(a_0,a_T)$ \\
\midrule
1  & 0.6759 & 0.3060 & 0.5691 & 0.0569 & 0.0181 \\
2  & 0.6918 & 0.4472 & 0.3762 & 0.0376 & 0.0271 \\
3  & 0.7027 & 0.5272 & 0.2700 & 0.0270 & 0.0324 \\
5  & 0.7171 & 0.6126 & 0.1607 & 0.0161 & 0.0386 \\
10 & 0.7349 & 0.6915 & 0.0667 & 0.0067 & 0.0450 \\
20 & 0.7484 & 0.7333 & 0.0233 & 0.0023 & 0.0490 \\
40 & 0.7573 & 0.7526 & 0.0072 & 0.0007 & 0.0512 \\
\bottomrule
\end{tabular}
\end{table}

\subsection{Bounded wages}

To illustrate feasibility, set $\kappa=0.005, \wbar=0.20.$ Figure~\ref{fig:bonuses} plots the minimum success bonuses. The flat line is the wage cap; a record can implement high effort only where its bonus curve lies below it. The outcome curve starts well under the cap and climbs, crossing it after $T=3$. At $T=1$, the outcome contract is feasible, but the agreement contract is not. At $T=2$, both are feasible, and the outcome contract still requires the smaller bonus. At $T=3$, the agreement contract becomes cheaper in incentive terms. From $T=4$ onward, the outcome bonus exceeds the wage cap, while the agreement bonus remains below it and converges to approximately $0.0931$.

\begin{figure}[htbp!]
\centering
\includegraphics[width=0.78\textwidth]{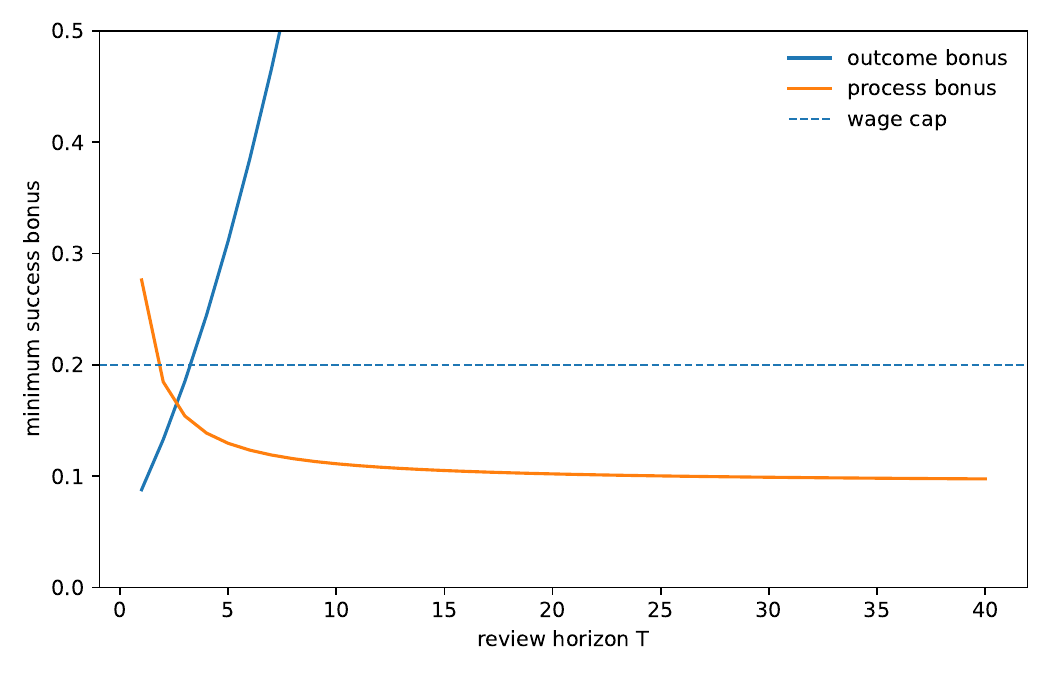}
\caption{Minimum success bonuses for $\kappa=0.005$. The dashed line is the wage cap $\wbar=0.20$.}
\label{fig:bonuses}
\end{figure}

All figures are generated directly from the recursion in Proposition~\ref{prop:hidden-recursion}. 

\section{Beyond memory-1}\label{sec:beyond}

The memory-1 benchmark buys one thing: a one-dimensional belief recursion, and with it the monotonicity, crossing, and washout-rate results. The branch decomposition applies far more widely. This section states which conclusions hold for arbitrary sign-symmetric architectures and which the benchmark alone supplies.

Fix an arbitrary sign-symmetric public architecture and an arbitrary continuation strategy profile that is held fixed across the hidden-effort comparison. Reviewers may observe the full action history, a network neighborhood, or another public state. The objects $S_T$ and $R_T$ continue to summarize the two endpoint records, but their evolution with the horizon now depends on the architecture.

\subsection{Endpoint geometry}\label{subsec:universal}

Continue to define
\[
  S_T:=\PP(a_T=\theta\mid a_0=\theta,\sigH),
  \qquad
  R_T:=\PP(a_T=\theta\mid a_0\ne\theta,\sigH).
\]
No Markov or memory restriction is imposed.

\begin{proposition}[Universal endpoint formulas]\label{prop:universal}
For every sign-symmetric downstream architecture satisfying Assumption~\ref{ass:exogeneity},
\begin{align}
  D(\theta,a_T)&=\Del|S_T-R_T|,\label{eq:univ-out}\\
  D(a_0,a_T)&=\Del|S_T+R_T-1|.\label{eq:univ-agree}
\end{align}
If
\begin{equation}\label{eq:positive-orientation}
  S_T>R_T,
  \qquad
  S_T+R_T>1,
\end{equation}
then process evaluation has weakly greater distinguishability than outcome evaluation if and only if
\begin{equation}\label{eq:universal-repair-threshold}
  R_T\ge\frac12.
\end{equation}
The same threshold ranks the required bonuses and, whenever both contracts are feasible, the minimum expected wage costs.
\end{proposition}

The repair threshold therefore rests on the two-branch structure and sign symmetry alone; the memory-1 benchmark contributes the horizon dynamics of $(S_T,R_T)$.

Outside the positive-orientation region, the favorable event may be terminal error or first--last disagreement. The relevant strengths are then $|S_T-R_T|$ and $|S_T+R_T-1|$. The formulas remain valid, but the simple one-half ranking must be stated with the appropriate signs.

\subsection{Full-history Case}\label{subsec:full-history}

For the full action history $H_T=(a_0,\ldots,a_T)$, Theorem~\ref{thm:decomposition} gives
\begin{equation}\label{eq:full-history-D-again}
  D(H_T)
  =\Del W_T,
\end{equation}
where
\begin{equation}\label{eq:WT}
  W_T
  :=\TV\!\left(
    \cL(H_T\mid a_0=\theta,\sigH),
    \cL(H_T\mid a_0\ne\theta,\sigH)
  \right)\in[0,1].
\end{equation}
The coefficient $W_T$ is the best branch separation available from the complete public trace. By data processing,
\begin{equation}\label{eq:endpoint-below-full}
  D(\theta,a_T)\le D(\theta,H_T),
  \qquad
  D(a_0,a_T)\le D(H_T),
\end{equation}
with the first comparison interpreted in the relevant state-observable information environment.

A rich history can matter even when an endpoint statistic is weak. For example, a sequence containing an early reversal followed by stable actions may reveal more about the initial branch than the final action alone. Proposition~\ref{prop:finite-contract} describes how a principal with a finite history record allocates payments across such histories. The agreement contract earns its place through transparency and a two-action archival requirement; richer records can strictly improve on it.

\begin{corollary}[When both branches eventually reach the truth]\label{cor:asym-learning}
Suppose an arbitrary architecture satisfies
\begin{equation}\label{eq:both-branches-learn}
  S_T\to1,
  \qquad
  R_T\to1.
\end{equation}
Then
\begin{align}
  D(\theta,a_T)&\to0,\label{eq:learn-out-zero}\\
  D(a_0,a_T)&\to\Del,\label{eq:learn-agree-Delta}\\
  D(H_T)&\to\Del.\label{eq:learn-full-Delta}
\end{align}
Hence first--last agreement asymptotically exhausts the hidden-effort information in the full public history.
\end{corollary}

The result illustrates the difference between decision quality and attribution. Perfect downstream correction makes the terminal action an excellent decision but a useless measure of which effort level produced the initial report. At the same time, perfect correction turns agreement into an exact indicator of whether the initial report was correct: good reports survive and bad reports are reversed.

\subsection{Bounded full-history cascade}\label{subsec:cascade}

Full-history observation with bounded signals can keep both endpoint records alive in the limit.

\begin{example}[Binary full-history cascade]\label{ex:cascade}
Suppose $\alpha_H=q\in(1/2,1)$ and $\alpha_L<q$. Each downstream reviewer observes the full public action history and one binary private signal of accuracy $q$. Reviewers believe that the initial report also has accuracy $q$. Ties are broken in favor of the reviewer's private signal.

This is the standard two-signal cascade geometry. After two identical public actions, a strict cascade begins. Let $\Gamma$ be the probability that the eventual cascade is correct when public evidence is neutral. Starting from neutrality, the next two informative actions are both correct with probability $q^2$, both incorrect with probability $(1-q)^2$, and split with probability $2q(1-q)$, which returns the process to neutrality. Therefore
\begin{equation}\label{eq:Gamma-recursion}
  \Gamma=q^2+2q(1-q)\Gamma,
\end{equation}
so
\begin{equation}\label{eq:Gamma}
  \Gamma=\frac{q^2}{q^2+(1-q)^2}.
\end{equation}

Conditional on a correct initial report, the next reviewer confirms it with probability $q$ and otherwise returns the process to neutrality. Conditional on an incorrect initial report, a correct private signal, which occurs with probability $q$, returns the process to neutrality. Hence
\begin{equation}\label{eq:cascade-SR}
  S_\infty=q+(1-q)\Gamma,
  \qquad
  R_\infty=q\Gamma.
\end{equation}
Substituting into Proposition~\ref{prop:universal} gives
\begin{align}
  \lim_{T\to\infty}D(\theta,a_T)
  &=\Del\frac{q(1-q)}{q^2+(1-q)^2}>0,\label{eq:cascade-out}\\
  \lim_{T\to\infty}D(a_0,a_T)
  &=\Del\frac{(2q-1)(q^2-q+1)}{q^2+(1-q)^2}>0.\label{eq:cascade-agree}
\end{align}
Both endpoint records survive because the eventual cascade remains dependent on the initial branch.
\end{example}

The repair threshold still determines the ranking. In the cascade example,
\[
  R_\infty
  =q\Gamma
  =\frac{q^3}{q^2+(1-q)^2}.
\]

\begin{corollary}[Asymptotic ranking in the cascade example]\label{cor:cascade-threshold}
There is a unique $q^\star\in(1/2,1)$ satisfying
\begin{equation}\label{eq:qstar-equation}
  2(q^\star)^3-2(q^\star)^2+2q^\star-1=0.
\end{equation}
Numerically, $q^\star\approx0.647799$. If $q>q^\star$, process evaluation is asymptotically stronger, requires a smaller wage cap, and is cheaper whenever both contracts are feasible. If $1/2<q<q^\star$, outcome evaluation has the corresponding advantages.
\end{corollary}

The example locates the boundary of the memory-1 conclusions: the decomposition into $(S_T,R_T)$ and the one-half repair threshold carry over, while the vanishing of $S_T-R_T$ in the contestable regime is supplied by the memory-1 dynamics.

The results can be summarized as follows. Three statements are architecture-free under the maintained exogeneity and symmetry assumptions:
\begin{enumerate}[label=(\roman*)]
\item every record's incentive power is the primitive effort effect times branch separation;
\item the canonical endpoint distances are determined by $(S_T,R_T)$;
\item in the positive-orientation region, the repair probability $R_T$ ranks the two endpoint contracts.
\end{enumerate}
By contrast, the evolution of $S_T$, $R_T$, and the full-history coefficient $W_T$ depends on who observes what, which signals are bounded, and how actions reveal private information. Those are the margins on which alternative organizational architectures should be compared.

\section{Retained traces and audited nodes}\label{sec:trace}

The benchmark uses the first--last agreement because that is one of the smallest records to have the attribution force present. This section explains how to read that benchmark in a broader spectrum. The organization may retain a randomly selected earlier recommendation, or it may audit a node other than the first analyst. The formulas below are local incentive formulas: given an intended high-effort continuation, they measure how much incentive power a retained trace gives to the audited information producer, one deviation at a time.

\subsection{Random trace records}\label{subsec:random-trace}

Let $J$ be an audit index taking values in a subset of $\{0,1,\ldots,T\}$, with distribution $\mu$, independent of the state, effort, and downstream private signals. Suppose the hidden-truth settlement record is
\begin{equation}\label{eq:random-trace-record}
  G_T^\mu:=(J,a_J,a_T).
\end{equation}
The case $\mu(0)=1$ is the first--last pair record. A distribution with support on several dates represents an organization that always observes the final recommendation but retains, or samples, one earlier recommendation from the chain.

\begin{proposition}[Random trace decomposition]\label{prop:random-trace}
For the random trace record in \eqref{eq:random-trace-record},
\begin{equation}\label{eq:random-trace-D}
  D(G_T^\mu)
  =
  \Del
  \sum_j \mu(j)
  \TV\!\left(
    \cL(a_j,a_T\mid a_0=\theta),
    \cL(a_j,a_T\mid a_0\ne\theta)
  \right).
\end{equation}
Consequently,
\begin{equation}\label{eq:random-trace-bound}
  D(G_T^\mu)
  \le
  \Del.
\end{equation}
If $\mu(0)=1$, \eqref{eq:random-trace-D} reduces to the pair-record branch distance studied below.
\end{proposition}

\begin{proof}
Apply Theorem~\ref{thm:decomposition} to $G_T^\mu$. Conditional on either initial branch, the audit index $J$ has the same distribution $\mu$ and is observed as part of the record. Because the realized audit date is part of the record, any event in the record space decomposes across the disjoint cells $\{J=j\}$, and within each cell the branch-separating event can be chosen freely; the total variation distance between the two conditional laws of $(J,a_J,a_T)$ therefore decomposes as
\[
  \TV\!\left(\cL(J,a_J,a_T\mid C),\cL(J,a_J,a_T\mid M)\right)
  =
  \sum_j \mu(j)
  \TV\!\left(\cL(a_j,a_T\mid C),\cL(a_j,a_T\mid M)\right),
\]
where $C=\{a_0=\theta\}$ and $M=\{a_0\ne\theta\}$. Multiplying by $\Del$ gives \eqref{eq:random-trace-D}. The bound follows because each total variation distance is at most one and the audit probabilities sum to one.
\end{proof}

The formula clarifies why the initial--terminal record should not be interpreted as the only possible process record. It is the degenerate audit that always opens the file at date zero. A richer organization might audit an intermediate recommendation, a random earlier recommendation, or several recommendations. The same criterion applies: the retained trace is valuable only to the extent that it separates histories generated after a correct initial report from histories generated after a mistaken one.

\subsection{Local accountability for an audited node}\label{subsec:local-node}

Now fix an arbitrary date $i<T$ and suppose agent $i$ is the evaluated information producer. Let $H_i$ denote the public history observed before agent $i$ acts. Conditional on $H_i=h$, agent $i$ chooses effort $e_i\in\{H,L\}$, and effort changes the probability that her recommendation matches the state:
\begin{equation}\label{eq:local-alpha}
  \PP(a_i=\theta\mid H_i=h,e_i=e)=\alpha_i^e(h),
  \qquad
  \Delta_i(h):=\alpha_i^H(h)-\alpha_i^L(h)\ge0.
\end{equation}
The downstream continuation strategy is fixed at the intended high-effort profile. The local analogue of Assumption~\ref{ass:exogeneity} is that, conditional on $(H_i,\theta,a_i)$, the future trace is independent of $e_i$.

For a future record $G_{i,T}$, define the conditional branch laws
\begin{align}
  \nu_{i,G}^{C,h}&:=\cL(G_{i,T}\mid H_i=h,
a_i=\theta),\label{eq:local-branch-C}\\
  \nu_{i,G}^{M,h}&:=\cL(G_{i,T}\mid H_i=h,
a_i\ne\theta).\label{eq:local-branch-M}
\end{align}

\begin{proposition}[Local branch decomposition]\label{prop:local-branch}
For each realized pre-action history $h$ and each effort $e$,
\begin{equation}\label{eq:local-mixture}
  \cL(G_{i,T}\mid H_i=h,e_i=e)
  =
  \alpha_i^e(h)\nu_{i,G}^{C,h}
  +
  \bigl(1-
  \alpha_i^e(h)\bigr)\nu_{i,G}^{M,h}.
\end{equation}
Thus the conditional incentive distance is
\begin{equation}\label{eq:local-D}
  D_i(G_{i,T}\mid h)
  =
  \Delta_i(h)\,
  \TV\!\left(\nu_{i,G}^{C,h},\nu_{i,G}^{M,h}\right).
\end{equation}
If the settlement record includes $H_i$, and the distribution of $H_i$ is not affected by agent $i$'s effort, then
\begin{equation}\label{eq:local-integrated-D}
  D_i(H_i,G_{i,T})
  =
  \EE\!\left[
    \Delta_i(H_i)
    \TV\!\left(\nu_{i,G}^{C,H_i},\nu_{i,G}^{M,H_i}\right)
  \right].
\end{equation}
\end{proposition}

\begin{proof}
Fix $h$ with $\PP(H_i=h)>0$. Conditional on $(H_i,\theta,a_i)$, the future trace is independent of $e_i$ by the local exogeneity assumption, and the continuation profile is fixed; hence the conditional law of $G_{i,T}$ given $\{H_i=h,\,a_i=\theta\}$, and likewise given $\{H_i=h,\,a_i\ne\theta\}$, does not depend on $e_i$. These are the branch laws \eqref{eq:local-branch-C}--\eqref{eq:local-branch-M}, and effort enters the conditional law of $G_{i,T}$ given $H_i=h$ only through the weights $\alpha_i^e(h)$ and $1-\alpha_i^e(h)$, which is \eqref{eq:local-mixture}. For any event $B$ in the record space, \eqref{eq:local-mixture} gives
\[
  \PP(G_{i,T}\in B\mid H_i=h,e_i=H)-\PP(G_{i,T}\in B\mid H_i=h,e_i=L)
  =
  \Delta_i(h)\bigl[\nu_{i,G}^{C,h}(B)-\nu_{i,G}^{M,h}(B)\bigr],
\]
and taking the supremum over $B$ yields \eqref{eq:local-D}. Finally, suppose the settlement record is the pair $(H_i,G_{i,T})$ and the law of $H_i$ is the same under both effort levels; write $\mu$ for that common law and recall that $H_i$, a finite history of binary actions, takes finitely many values. Any event $B'$ in the pair-record space decomposes across the cells $\{H_i=h\}$ with $h$-sections $B'_h$, so
\[
  \PP\bigl((H_i,G_{i,T})\in B'\mid e_i=H\bigr)-\PP\bigl((H_i,G_{i,T})\in B'\mid e_i=L\bigr)
  =
  \sum_h \mu(h)\,\Delta_i(h)\bigl[\nu_{i,G}^{C,h}(B'_h)-\nu_{i,G}^{M,h}(B'_h)\bigr].
\]
Each summand is maximized separately by choosing $B'_h$ to attain $\TV(\nu_{i,G}^{C,h},\nu_{i,G}^{M,h})$, and those choices assemble into a single event $B'$. The supremum of the left side over $B'$ therefore equals the right side of \eqref{eq:local-integrated-D}.
\end{proof}

Proposition~\ref{prop:local-branch} is the object needed to discuss multiple evaluated agents without choosing among all equilibria. Fix an intended profile in which the relevant agents exert high effort. For each agent $i$, the principal checks whether the record assigned to that agent generates enough local distance to cover $\kappa_i$. If agent $i$ is audited with probability $\pi_i$, and the record is otherwise uninformative about $e_i$, the maximal incentive gap is multiplied by $\pi_i$:
\begin{equation}\label{eq:audited-IC}
  \pi_i\wbar\,D_i(H_i,G_{i,T})\ge \kappa_i.
\end{equation}
Equation~\eqref{eq:audited-IC} is the standard one-deviation requirement for supporting the intended high-effort profile, holding the downstream continuation fixed; uniqueness of implementation is a separate question left open.

\begin{remark}[Endpoint formulas for later nodes]
For the initial analyst, symmetry reduces the endpoint records $(\theta,a_T)$ and $(a_0,a_T)$ to the simple repair probabilities studied in Section~\ref{sec:hidden}. For a later audited node, the public history $H_i$ may contain an asymmetric public belief, so the optimal process score need not be literal agreement. The robust object is the branch distance in \eqref{eq:local-D}. In symmetric local subproblems the same formulas reappear with $a_i$ replacing $a_0$; outside symmetry the bounded-wage optimal contract ranks the cells of $(H_i,a_i,a_T)$ by their likelihood ratio under high versus low effort.
\end{remark}

\section{Asymmetric priors and directional contestability}\label{sec:asymmetry}

The symmetric benchmark makes first--last agreement a sufficient process statistic. With an asymmetric prior, the sign of the initial report itself carries effort information, and the two report directions can differ in how contestable they are. The branch logic carries over; the canonical process score becomes a likelihood-ratio ranking of the four endpoint cells, with pooled agreement as its symmetric special case.

This section keeps the binary state and the memory-1 information structure but allows the prior to favor one state. Let
\begin{equation}\label{eq:asym-prior}
  p:=\PP(\theta=+1),
  \qquad
  \omega:=\log\frac{p}{1-p}
\end{equation}
be the prior log-odds. Under the high-effort belief, the initial report has log-likelihood ratio
\begin{equation}\label{eq:asym-mH}
  m_H:=\log\frac{\alpha_H}{1-\alpha_H}.
\end{equation}
Thus the public log-odds after report $+1$ and after report $-1$ are
\begin{equation}\label{eq:asym-uv}
  u_0:=\omega+m_H,
  \qquad
  v_0:=\omega-m_H.
\end{equation}
The two post-report beliefs are shifted together by the prior, while their distance $u_0-v_0=2m_H$ is the perceived informativeness of the initial report under the high-effort continuation.

The private signal support creates two frontiers, $-\Lb$ and $\Lb$. A post-report belief inside $(-\Lb,\Lb)$ is contestable: a sufficiently strong private signal can overturn it. A belief above $\Lb$ forces the next reviewer to choose $+1$, and a belief below $-\Lb$ forces the next reviewer to choose $-1$. Figure~\ref{fig:asym-regimes} summarizes the six possible positions of the ordered pair $(u_0,v_0)$.

\begin{figure}[htbp!]
\centering
\includegraphics[width=0.92\textwidth]{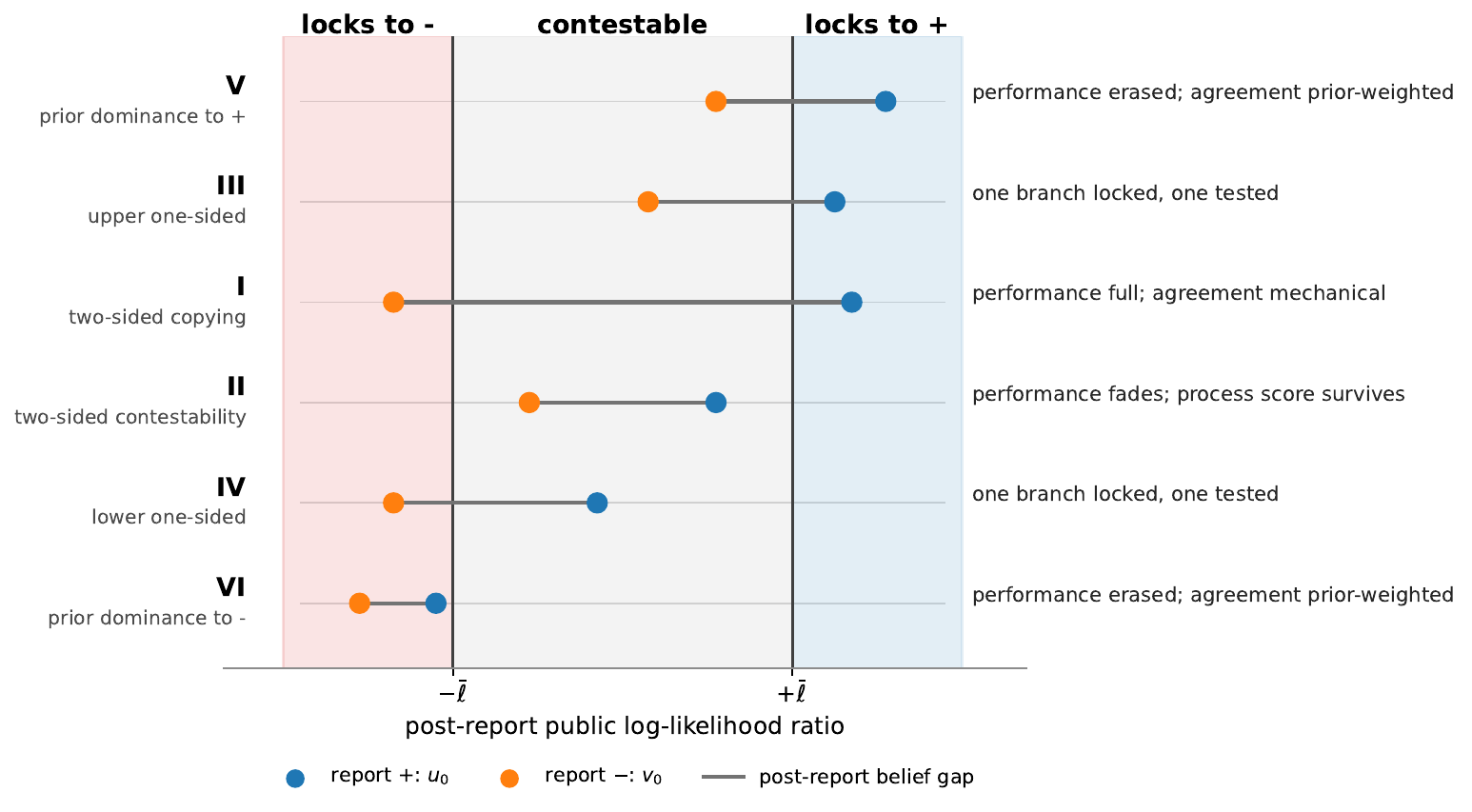}
\caption{Directional contestability under asymmetric priors. The horizontal axis is the public log-odds after observing the initial report. The two frontiers are the maximum private-signal strengths. Since $u_0>v_0$, the ordered pair can occupy only six of the nine cells generated by the three regions.}
\label{fig:asym-regimes}
\end{figure}

\begin{proposition}[Initial directional regimes]\label{prop:asym-regimes}
Assume $u_0>v_0$ and that private signals have support $[-\Lb,\Lb]$ with no atoms at the frontiers. Exactly one of the following six cases holds.

\begin{enumerate}[label=(\roman*),leftmargin=2.2em]
\item \emph{Two-sided copying.} If $u_0\ge\Lb$ and $v_0\le-\Lb$, report $+1$ is uncontestable upward and report $-1$ is uncontestable downward. The first report determines the downstream action direction. Terminal correctness preserves the initial-report branch, while first--last agreement is mechanical.

\item \emph{Two-sided contestability.} If $-\Lb<v_0<u_0<\Lb$, both report directions are contestable. In the symmetric case, this is the frontier regime studied above. In asymmetric chains that remain in the interior and converge to the two frontiers, terminal correctness loses first-report attribution direction by direction, while process cells retain directional branch information.

\item \emph{Upper one-sided contestability.} If $u_0\ge\Lb$ and $v_0\in(-\Lb,\Lb)$, a positive report is locked in but a negative report is tested. Attribution is partial and directional.

\item \emph{Lower one-sided contestability.} If $v_0\le-\Lb$ and $u_0\in(-\Lb,\Lb)$, a negative report is locked in but a positive report is tested. Attribution is again partial and directional.

\item \emph{Prior dominance toward state $+1$.} If $v_0\ge\Lb$, even a negative report leaves the public belief above the upper frontier. From the first downstream review on, both reports are pushed to action $+1$. Terminal correctness cannot distinguish the two initial-report directions, and agreement mostly rewards conformity with the prior-favored action rather than survival of scrutiny.

\item \emph{Prior dominance toward state $-1$.} If $u_0\le-\Lb$, even a positive report leaves the public belief below the lower frontier. From the first downstream review on, both reports are pushed to action $-1$. Terminal correctness again collapses the two initial-report directions, and agreement is prior-contaminated.
\end{enumerate}
\end{proposition}

The classification highlights the difference between copying and prior dominance. In two-sided copying, the two reports are locked into opposite actions, so terminal correctness retains the initial-report branch. In prior dominance, both reports are locked into the same action, so terminal correctness erases that branch. Both regimes make private signals irrelevant, but they have opposite implications for attribution.

The contracting object is the endpoint record itself, cell by cell. Let
\begin{equation}\label{eq:asym-endpoint-record}
  G_T^{A}:=(a_0,a_T)
\end{equation}
be the process endpoint record and let
\begin{equation}\label{eq:asym-cell-delta}
  \delta_{st}:=\PP^H(a_0=s,a_T=t)-\PP^L(a_0=s,a_T=t),
  \qquad s,t\in\{-1,+1\}.
\end{equation}
The sign and magnitude of $\delta_{st}$ tell the principal whether a cell is favorable or unfavorable evidence about high effort.

\begin{proposition}[Cell-specific process scores]\label{prop:asym-cells}
For the finite endpoint record $G_T^A=(a_0,a_T)$, the maximum bounded incentive gap equals
\begin{equation}\label{eq:asym-cell-gap}
  \wbar\sum_{s,t:\,\delta_{st}>0}\delta_{st}.
\end{equation}
A gap-maximizing bounded contract pays the cap on cells with $\delta_{st}>0$ and zero on cells with $\delta_{st}<0$. A minimum-cost bounded contract pays first on positive cells with the smallest ratio
\begin{equation}\label{eq:asym-cell-ratio}
  \frac{\PP^H(a_0=s,a_T=t)}{\delta_{st}}.
\end{equation}
\end{proposition}

\begin{proof}
Apply Proposition~\ref{prop:finite-contract} to the four cells of $G_T^A$. For any set $B$ of endpoint cells,
\[
  \PP^H(G_T^A\in B)-\PP^L(G_T^A\in B)
  =\sum_{(s,t)\in B}\delta_{st},
\]
and because the four $\delta_{st}$ sum to zero, the difference is maximized by collecting exactly the cells with $\delta_{st}>0$; scaling by the cap gives the maximum bounded incentive gap \eqref{eq:asym-cell-gap}, attained by paying $\wbar$ on those cells and zero elsewhere. The minimum-cost contract follows from the fractional-knapsack argument of Proposition~\ref{prop:finite-contract}: among positive cells, pay first on those with the lowest high-effort expected cost per unit of incentive, the ratio in \eqref{eq:asym-cell-ratio}.
\end{proof}

Under symmetry, the two agreement cells have the same sign and the two disagreement cells have the opposite sign, so the likelihood-ratio ranking collapses to the event $a_0=a_T$. Under asymmetric priors, the four cells need not collapse in this way. A contrarian report that survives review may be stronger evidence of effort than a popular report that survives review; a popular report that is overturned may be especially unfavorable. Agreement is thus the symmetric special case of a general cell-by-cell process score.

\section{Discussion}\label{sec:discussion}

A review process improves the final decision by correcting mistaken initial reports, and the better it does this, the less it reveals about the analyst who filed the first one. As the repair probability $R_T$ rises, the correct and mistaken rates converge to the same terminal accuracy, so the outcome record no longer reflects the analyst's hidden effort. The same rise makes first--last
agreement informative, since a mistaken report is now unlikely to survive review. Organizational correction and individual attribution move in opposite directions, and this is what turns record-keeping into part of incentive design: which trace an organization retains, and which one it pays on, remains open even once the quality of review is fixed. The rest of this section develops that point as an empirical diagnostic, under behavioral reviewers, and at the boundaries of the model.

\subsection{An empirical diagnostic}

The branch representation suggests a direct empirical diagnostic when the state is eventually observed in a validation sample. For a fixed review architecture and horizon, one can estimate
\[
  S_T=\PP(a_T=\theta\mid a_0=\theta),
  \qquad
  R_T=\PP(a_T=\theta\mid a_0\ne\theta)
\]
from cases in which both the initial and terminal actions can be linked to a realized state. These estimates immediately imply the relative endpoint strengths
\[
  |S_T-R_T|
  \qquad\text{and}\qquad
  |S_T+R_T-1|,
\]
without estimating the full structural learning model. In the positive-orientation region, checking whether $R_T$ exceeds one half identifies which endpoint record requires the smaller success bonus.

Estimating the absolute implementation requirement additionally needs the effort effect $\Del$ and the incremental cost $\kappa$. Those objects may come from an effort intervention, a change in investigation resources, or a calibrated production technology. The formulas are then operational:
\[
  \wbar_{\min}^\theta
  =\frac{\kappa}{\Del|S_T-R_T|},
  \qquad
  \wbar_{\min}^A
  =\frac{\kappa}{\Del|S_T+R_T-1|}.
\]
Carrying them across architectures requires re-estimating the branch probabilities, which are equilibrium objects under the specified continuation rule.

\subsection{Behavioral social learning}
The benchmark assumes reviewers read predecessors' actions with correct Bayesian logic. Laboratory evidence is more nuanced: \citet{AngrisaniGuarinoJehielKitagawa2021} find that subjects weight their own signal correctly but discount the information in predecessors' actions, a pattern they read as relative overconfidence in one's own signal over others'. 

Such behavior amounts to a fixed continuation rule with a constant discount on inherited evidence, and the branch decomposition prices it through its effect on the preservation and repair probabilities $(S_T, R_T)$. Discounting predecessors makes reviewers more likely to override an inherited action when they receive their own signal, which raises the repair probability $R_T$. By Theorem~\ref{thm:repair-threshold}, higher $R_T$ makes the first--last agreement contract dominate and accelerates the washing-out of outcome-based attribution. Such downstream overconfidence pushes organizations toward a regime in which reducing herding improves final decisions while eroding, faster, the record that ties those decisions to early effort. The divergence between decision quality and attribution extends beyond the Bayesian benchmark.

\subsection{Scope}\label{subsec:scope}
The results hold within four boundaries. First, the initial report is the recorded output of the investigation technology; a model in which the analyst could revise it after seeing her signal would need honest reporting to be incentive-compatible as well, bringing in scoring rules and elicitation. Second, once the state and initial report are fixed, the downstream record is the same under high and low effort, which rules out side traces such as raw notes or confidence scores that would themselves reveal effort; admitting them calls for a multibranch version of Theorem~\ref{thm:decomposition}. Third, reviewers are taken as given---they match the state and receive their information exogenously---so endogenizing review effort and the review structure around it is left open, though Section~\ref{sec:trace} does hold one audited reviewer to an incentive constraint. Fourth, the benchmark is binary and symmetric, which is what buys the two-number summary $(S_T, R_T)$; under an asymmetric prior or signal, the direction of the initial report itself signals effort and the best wage varies across endpoint cells, so the branch decomposition continues to apply while the two-number representation gives way to the likelihood-ratio cell ranking that Section~\ref{sec:asymmetry} develops for the prior.

\section{Conclusion}\label{sec:conclusion}

This paper places a hidden-action problem at the head of a sequential review chain and works out what the chain's records can still say about the effort behind the first report. The central tool is an accounting identity: every downstream record earns its incentive value from the separation it preserves between review histories that began with a correct report and histories that began with a mistaken one. For the two canonical settlement records, that separation reduces to the preservation and repair probabilities $(S_T,R_T)$, and the repair probability alone ranks the records---once review repairs a mistaken report more often than not, paying on first--last agreement requires a smaller bonus than paying on terminal correctness.

The memory-1 benchmark converts this accounting into horizon dynamics. When the initial report is contestable, public confidence climbs toward the copying frontier, the outcome bonus grows quadratically in the review horizon until any fixed wage cap makes outcome evaluation infeasible, and the agreement bonus falls toward a finite limit; the preferred endpoint record switches exactly once, at the horizon where the repair probability crosses one half. When the initial report is strong enough to be copied outright, the same accounting reverses the ranking: the terminal action inherits the full effort imprint and agreement is degenerate.

Beyond the benchmark, the same branch decomposition organizes each extension. Retained traces and audited intermediate nodes are priced by the local version of the identity, one deviation at a time; reviewers who discount their predecessors enter through the preservation and repair probabilities they induce; and asymmetric priors replace pooled agreement with a likelihood-ratio ranking of endpoint cells, of which agreement is the symmetric special case. The objects at the center of the analysis---$S_T$, $R_T$, and the primitive effort effect $\Del$---are conditional probabilities that a validation sample can estimate, so the theory doubles as a diagnostic for choosing what a review process should archive and what it should pay on.

\appendix
\begin{appendices}

\section{Proofs for Section~\ref{sec:records}}\label{app:record-proofs}

\begin{proof}[Proof of Proposition~\ref{prop:maxgap}]
Let $P:=\cL^H(G)$ and $Q:=\cL^L(G)$ on the record space $\Omega_G$, and let $\eta:=P-Q$, a finite signed measure with $\eta(\Omega_G)=0$. By the Hahn decomposition there is a measurable set $A\subseteq\Omega_G$ on whose subsets $\eta$ is nonnegative and on whose complement's subsets $\eta$ is nonpositive; the Jordan parts are $\eta^+(\cdot)=\eta(\cdot\cap A)$ and $\eta^-(\cdot)=-\eta(\cdot\cap A^c)$. Because $\eta(\Omega_G)=0$, the two parts have equal total mass, and for any event $B$,
\[
  \eta(B)\le\eta(B\cap A)\le\eta(A)=\eta^+(\Omega_G),
\]
with equality at $B=A$; since $\eta(B^c)=-\eta(B)$, the same value bounds $|\eta(B)|$, so
\[
  \eta^+(\Omega_G)=\eta^-(\Omega_G)=\sup_B\,|P(B)-Q(B)|=\TV(P,Q).
\]
Now take any wage $w$ with $0\le w\le\wbar$. Then
\[
  \EE^H[w(G)]-\EE^L[w(G)]
  =\int w\,d\eta
  =\int w\,d\eta^+-\int w\,d\eta^-
  \le\wbar\,\eta^+(\Omega_G)-0
  =\wbar\,\TV(P,Q),
\]
and the wage $w^*:=\wbar\1_A$ attains the bound. This proves \eqref{eq:maxgap}: the gap-maximizing contract pays the cap exactly on the region of record values that are relatively more likely after high effort and pays zero elsewhere. Since the incentive constraint \eqref{eq:IC} requires the gap to reach $\kappa$, and the supremum of the gap over bounded wages is attained, some bounded wage on $G$ implements high effort if and only if $\wbar D(G)\ge\kappa$, which is \eqref{eq:feasibility}.
\end{proof}

\begin{proof}[Proof of Proposition~\ref{prop:finite-contract}]
A wage on $G$ is a vector $(w_i)_{i\le n}$ with $0\le w_i\le\wbar$; its incentive gap is $\sum_i w_i\delta_i$ and its high-effort expected cost is $\sum_i w_i p_i^H$. In any minimum-cost wage, $w_i=0$ whenever $\delta_i\le0$: lowering such a coordinate weakly raises the gap and weakly lowers cost, strictly so unless $p_i^H=0$, in which case the coordinate can be set to zero without affecting either side. The problem is therefore the linear program
\[
  \min_{0\le w_i\le\wbar}\ \sum_{i\in I_+}w_i\,p_i^H
  \qquad\text{subject to}\qquad
  \sum_{i\in I_+}w_i\,\delta_i\ \ge\ \kappa.
\]
For feasibility, note that for any set $B$ of cells, $\PP^H(G\in B)-\PP^L(G\in B)=\sum_{i\in B}\delta_i$, and because $\sum_{i\le n}\delta_i=0$, the sum is maximized by collecting exactly the positive cells; hence the maximal bounded gap is $\wbar\sum_{i\in I_+}\delta_i=\wbar D(G)$, and the program is feasible if and only if \eqref{eq:feasibility} holds.

Substituting $y_i:=w_i\delta_i$ for $i\in I_+$ turns the program into a continuous knapsack: minimize $\sum_i\rho_i y_i$ subject to $\sum_i y_i\ge\kappa$ and $0\le y_i\le\wbar\delta_i$, with unit costs $\rho_i=p_i^H/\delta_i$. Order the cells so that $\rho_1\le\cdots\le\rho_m$. If a feasible solution left slack on some cell $j$ while paying on a cell $k$ with $\rho_k>\rho_j$, transferring incentive mass $\varepsilon>0$ from $k$ to $j$ would preserve the constraint and change cost by $\varepsilon(\rho_j-\rho_k)\le0$, strictly if $\rho_j<\rho_k$; iterating, some optimal solution fills capacity in increasing order of $\rho_i$ until the constraint binds. Translating back through $w_i=y_i/\delta_i$ gives \eqref{eq:finite-optimal-wage}. If several cells share the cutoff ratio $\rho_k$, transfers among them move cost at the common rate $\rho_k$ per unit of incentive, so payments can be redistributed across the tied cells without changing cost.
\end{proof}

\begin{proof}[Proof of Theorem~\ref{thm:decomposition}]
Every object verifiable at settlement is generated by the state, the initial report, and the downstream variables collected in $Z_T$, so the record can be written $G=g(\theta,a_0,Z_T)$ for a measurable function $g$. Fix an event $B$ in the record space and let $A_B:=g^{-1}(B)$. By Assumption~\ref{ass:exogeneity}, conditional on $(\theta,a_0)$ the law of $Z_T$ is the kernel $K(\cdot\mid\theta,a_0)$ under either effort level, so
\[
  \PP^e\bigl(G\in B\mid\theta,a_0\bigr)
  =K\bigl(\{z:(\theta,a_0,z)\in A_B\}\,\big|\,\theta,a_0\bigr)
\]
carries no dependence on $e$. Under the symmetric technology \eqref{eq:initial-tech}, the conditional law of the pair $(\theta,a_0)$ within either branch is also effort free: on $C$ it places probability one half on each of $(+1,+1)$ and $(-1,-1)$, and on $M$ probability one half on each of $(+1,-1)$ and $(-1,+1)$. Combining the two facts, $\PP^e(G\in B\mid C)=\nu_G^C(B)$ and $\PP^e(G\in B\mid M)=\nu_G^M(B)$ for both effort levels. Effort therefore enters the law of $G$ only through the branch weights:
\[
  \PP^e(G\in B)
  =\alpha_e\,\nu_G^C(B)+(1-\alpha_e)\,\nu_G^M(B),
\]
which is \eqref{eq:mixture}. Subtracting the two effort levels gives \eqref{eq:event-gap} for every event $B$, and taking the supremum over $B$ yields \eqref{eq:D-decomp}. The bound \eqref{eq:D-upper} follows because the total-variation distance between two probability measures is at most one.
\end{proof}

\begin{proof}[Proof of Proposition~\ref{prop:data-processing}]
Write $G'=g(G)$ for a measurable $g$. For any event $B$ in the range of $G'$, the preimage $A_B:=g^{-1}(B)$ satisfies $\{G'\in B\}=\{G\in A_B\}$, so
\[
  \bigl|\nu_{G'}^C(B)-\nu_{G'}^M(B)\bigr|
  =\bigl|\nu_{G}^C(A_B)-\nu_{G}^M(A_B)\bigr|
  \le\TV\!\left(\nu_G^C,\nu_G^M\right).
\]
Taking the supremum over $B$ gives the branch-distance comparison; multiplying by $\Del$ and using \eqref{eq:D-decomp} gives \eqref{eq:data-processing}. Every event the coarse record can price is an event the fine record already prices, so coarsening can only remove separating events.
\end{proof}

\section{Proofs for Section~\ref{sec:memory1}}\label{app:memory-proofs}

\begin{proof}[Proof of Lemma~\ref{lem:llr-identity}]
Let $m$ be a common reference measure for the signal densities $f_+$ and
$f_-$. Let $h_+$ and $h_-$ be the densities of the random variable $\ell_t$
under states $+1$ and $-1$, respectively. By definition,
\[
  \ell(x)=\log\frac{f_+(x)}{f_-(x)},
  \qquad\text{so}\qquad
  f_-(x)=e^{-\ell(x)}f_+(x).
\]
For every bounded measurable function $\varphi$,
\begin{align*}
  \EE[\varphi(\ell_t)\mid\theta=-1]
  &=
  \int \varphi(\ell(x)) f_-(x)\,m(dx)\\
  &=
  \int \varphi(\ell(x)) e^{-\ell(x)} f_+(x)\,m(dx)\\
  &=
  \EE[\varphi(\ell_t)e^{-\ell_t}\mid\theta=+1].
\end{align*}
Hence the law of $\ell_t$ under $\theta=-1$ is the law under $\theta=+1$
reweighted by $e^{-\ell_t}$. In density form,
\[
  h_-(z)=e^{-z}h_+(z)
\]
for almost every $z$. Since the densities are continuous on $(-\Lb,\Lb)$, the
identity holds for every $z$ in that interval. Sign symmetry gives
$h_-(z)=h_+(-z)$. Therefore, writing $h=h_+$,
\[
  h(-z)=h_+(-z)=h_-(z)=e^{-z}h_+(z)=e^{-z}h(z),
\]
as required.
\end{proof}

\begin{proof}[Proof of Lemma~\ref{lem:action-symmetry}]
All probabilities are under the on-path law $\PP^H$. The proof establishes, by induction on $t$, that
\[
  \PP^H(a_t=+1)=\tfrac12
  \qquad\text{and}\qquad
  \PP^H(\theta=a_t\mid a_t)=\hat r_t.
\]
At $t=0$, the symmetric prior and the symmetric technology \eqref{eq:initial-tech} give $\PP^H(a_0=+1)=\tfrac12$ and, by Bayes' rule,
\[
  \PP^H(\theta=+1\mid a_0=+1)
  =\frac{\alpha_H\cdot\tfrac12}{\tfrac12}
  =\alpha_H=\hat r_0,
\]
with the symmetric computation after $a_0=-1$.

For the inductive step, suppose the claim holds at date $t$. Reviewer $t+1$ observes $(a_t,\ell_{t+1})$, where $\ell_{t+1}$ is independent of $(a_t,a_0,\ldots)$ conditional on $\theta$, and chooses $a_{t+1}$ by the threshold rule \eqref{eq:action-rule} with threshold $s_t$ generated by $\hat r_t$. Consider the sign-flip map $(\theta,a_t,\ell_{t+1})\mapsto(-\theta,-a_t,-\ell_{t+1})$. It preserves the joint on-path law of $(\theta,a_t,\ell_{t+1})$: the prior is symmetric, the inductive hypothesis makes the law of $(\theta,a_t)$ invariant, and the experiment satisfies \eqref{eq:signal-symmetry}. The threshold rule \eqref{eq:action-rule} is odd under the map, so applying it flips $a_{t+1}$. Invariance of the law together with the flip of $a_{t+1}$ yields
\[
  \PP^H(a_{t+1}=+1)=\tfrac12
  \qquad\text{and}\qquad
  \PP^H(\theta=+1\mid a_{t+1}=+1)=\PP^H(\theta=-1\mid a_{t+1}=-1).
\]
Writing $x$ for the common conditional value and averaging over the two equally likely realizations of $a_{t+1}$,
\[
  \PP^H(\theta=a_{t+1})
  =\tfrac12\,x+\tfrac12\,x=x,
\]
so $x=\hat r_{t+1}$ by the definition \eqref{eq:rhat}, and \eqref{eq:posterior-correctness} holds at $t+1$.

Finally, for an observer whose information is $a_t$ alone, prior odds are even, so Bayes' rule gives log posterior odds for the event $\theta=+1$ equal to
\[
  \log\frac{\PP^H(\theta=+1\mid a_t)}{\PP^H(\theta=-1\mid a_t)}
  =a_t\log\frac{\hat r_t}{1-\hat r_t}
  =a_t s_t,
\]
which is the public log-likelihood ratio claimed in \eqref{eq:s-def}.
\end{proof}

\section{Proofs for Section~\ref{sec:hidden}}\label{app:hidden-proofs}

\begin{proof}[Proof of Proposition~\ref{prop:hidden-recursion}]
Under either actual effort, downstream reviewers use the threshold $s_t$ generated by the high-effort on-path belief path. Conditional on a correct predecessor, the next action is correct with probability $1-H(-s_t)$. Conditional on an incorrect predecessor, it is correct with probability $1-H(s_t)$. Therefore
\begin{align*}
  r_{t+1}(e)
  &=r_t(e)[1-H(-s_t)]
    +[1-r_t(e)][1-H(s_t)]\\
  &=1-H(s_t)+[H(s_t)-H(-s_t)]r_t(e)\\
  &=b_t+\lambda_t r_t(e).
\end{align*}
Subtracting the recursions for $H$ and $L$ gives
\[
  d_{t+1}=\lambda_t d_t.
\]
Since $d_0=\Del$, iteration gives \eqref{eq:d-product}.
\end{proof}

\begin{proof}[Proof of Corollary~\ref{cor:washout}]
Let $\delta_t:=\rb-\hat r_t$. Corollary~\ref{cor:regimes} gives $\delta_t\downarrow0$. Since $F(\rb)=\rb$, the mean value theorem gives
\[
  \delta_{t+1}
  =\rb-F(\hat r_t)
  =F'(\xi_t)\delta_t
\]
for some $\xi_t\in(\hat r_t,\rb)$. Because $F'$ is increasing,
\[
  \lambda_t=F'(\hat r_t)
  \le F'(\xi_t)
  =\frac{\delta_{t+1}}{\delta_t}.
\]
Multiplying from $t=0$ to $T-1$ yields
\[
  \Pi_T\le\frac{\delta_T}{\delta_0}\to0.
\]
Each $\lambda_t$ lies strictly between zero and one, so $\Pi_T$ is strictly decreasing. Equation \eqref{eq:d-product} gives the same conclusions for $d_T$.
\end{proof}

\begin{proof}[Proof of Proposition~\ref{prop:SR-memory1}]
Conditional on either initial branch, terminal correctness follows the affine recursion \eqref{eq:actual-recursion}. Let $S_t$ and $R_t$ denote the two branch paths at date $t$. Their initial values are
\[
  S_0=1,
  \qquad
  R_0=0,
  \qquad
  \hat r_0=\alpha_H.
\]
Subtracting the on-path recursion from the branch recursions gives
\[
  S_{t+1}-\hat r_{t+1}
  =\lambda_t(S_t-\hat r_t),
\]
and
\[
  R_{t+1}-\hat r_{t+1}
  =\lambda_t(R_t-\hat r_t).
\]
Iteration yields
\[
  S_T-\hat r_T=(1-\alpha_H)\Pi_T,
  \qquad
  R_T-\hat r_T=-\alpha_H\Pi_T,
\]
which are \eqref{eq:S} and \eqref{eq:R}. Adding and subtracting gives \eqref{eq:SminusR} and \eqref{eq:SplusR}.
\end{proof}

\begin{proof}[Proof of Theorem~\ref{thm:monotonicity}]
Part (i) follows because every $\lambda_t\in(0,1)$ and $\Pi_0=1$.

For part (ii), use \eqref{eq:R}:
\begin{align*}
  R_{T+1}-R_T
  &=(\hat r_{T+1}-\hat r_T)
    +\alpha_H(\Pi_T-\Pi_{T+1})>0.
\end{align*}
Both terms are strictly positive. Since $\hat r_T\to\rb$ and $\Pi_T\to0$, $R_T\to\rb$.

For part (iii), write \eqref{eq:SplusR} as
\[
  S_T+R_T-1
  =2\hat r_T-1-(2\alpha_H-1)\Pi_T.
\]
Hence
\begin{align*}
 &(S_{T+1}+R_{T+1}-1)-(S_T+R_T-1)\\
 &\qquad=2(\hat r_{T+1}-\hat r_T)
 +(2\alpha_H-1)(\Pi_T-\Pi_{T+1})>0.
\end{align*}
The initial value is zero, and the limit is $2\rb-1$.
\end{proof}

\begin{proof}[Proof of Proposition~\ref{prop:endpoint-D}]
Under effort $e$, terminal correctness occurs with probability
\[
  p_e^\theta
  =\alpha_e S_T+(1-\alpha_e)R_T.
\]
By sign symmetry, the two correct cells of $(\theta,a_T)$ each have probability $p_e^\theta/2$ and the two incorrect cells each have probability $(1-p_e^\theta)/2$. The total-variation distance between the high- and low-effort four-cell laws is therefore
\[
  |p_H^\theta-p_L^\theta|
  =\Del|S_T-R_T|.
\]

Agreement occurs with probability
\[
  p_e^A
  =\alpha_e S_T+(1-\alpha_e)(1-R_T),
\]
because agreement on the incorrect-initial-report branch means terminal error. Sign symmetry again reduces the four-cell pair law to agreement versus disagreement, so
\[
  D(a_0,a_T)
  =|p_H^A-p_L^A|
  =\Del|S_T+R_T-1|.
\]
The memory-1 formulas follow from Proposition~\ref{prop:SR-memory1} and positivity in Theorem~\ref{thm:monotonicity}.
\end{proof}

\begin{proof}[Proof of Proposition~\ref{prop:endpoint-costs}]
Under the positive orientation, the effort differences in the favorable-event probabilities are
\[
  p_H^\theta-p_L^\theta=\Del(S_T-R_T)
\]
and
\[
  p_H^A-p_L^A=\Del(S_T+R_T-1).
\]
Corollary~\ref{cor:binary-cost} gives the two bonuses and expected costs. On path,
\[
  p_H^\theta
  =\alpha_HS_T+(1-\alpha_H)R_T
  =\hat r_T,
\]
which gives the second expression in \eqref{eq:C-out}.

It remains to justify optimality on the original four-cell records. Under sign symmetry, the two correctness cells of $(\theta,a_T)$ have identical high- and low-effort probabilities, as do the two error cells. Averaging wages within each pair preserves both expected cost and the incentive gap. Thus an optimal wage may depend only on correctness. The same argument reduces $(a_0,a_T)$ to agreement versus disagreement. Corollary~\ref{cor:binary-cost} is therefore globally optimal on each endpoint record.
\end{proof}

\begin{proof}[Proof of Theorem~\ref{thm:repair-threshold}]
Using Proposition~\ref{prop:endpoint-D},
\begin{align*}
  D(a_0,a_T)-D(\theta,a_T)
  &=\Del[(S_T+R_T-1)-(S_T-R_T)]\\
  &=\Del(2R_T-1).
\end{align*}
This proves the distance ranking. Since the required bonuses are the inverses of the two positive incentive gaps, the same threshold gives the bonus ranking.

For expected costs, subtract \eqref{eq:C-out} from \eqref{eq:C-agree}. Direct simplification gives
\[
  C_T^A-C_T^\theta
  =-
  \frac{\kappa S_T(2R_T-1)}
       {\Del(S_T-R_T)(S_T+R_T-1)}.
\]
All factors other than $2R_T-1$ are positive, proving \eqref{eq:cost-ranking}.
\end{proof}

\begin{proof}[Proof of Corollary~\ref{cor:unique-switch}]
Theorem~\ref{thm:monotonicity} gives $R_0=0$, strict monotonicity, and $R_T\to\rb>1/2$. Hence the set in \eqref{eq:switch-horizon} is nonempty and has a finite minimum. The ranking before, at, and after that date follows from Theorem~\ref{thm:repair-threshold}.
\end{proof}

\begin{proof}[Proof of Theorem~\ref{thm:duality}]
In the frontier regime, Corollary~\ref{cor:regimes}, Corollary~\ref{cor:washout}, and Proposition~\ref{prop:SR-memory1} give
\[
  \hat r_T\to\rb,
  \qquad
  \Pi_T\to0,
  \qquad
  S_T,R_T\to\rb.
\]
Proposition~\ref{prop:endpoint-D} then yields \eqref{eq:out-dies} and \eqref{eq:agree-survives}. Proposition~\ref{prop:maxgap} gives the feasibility statements. In particular, process distinguishability is strictly below its limit at every finite horizon. Hence equality $\wbar\Del(2\rb-1)=\kappa$ still implies finite-horizon infeasibility.

In the immediate-copying regime, $a_T=a_0$ almost surely. Terminal correctness therefore has probability $\alpha_e$ under effort $e$, so its distance is $\Del$. Agreement occurs with probability one under both efforts, so its distance is zero.
\end{proof}

\section{Frontier rates and contract growth}\label{app:rates}

The main text uses only $\Pi_T\to0$. A polynomial rate requires regularity near the boundary of the LLR support.

\begin{assumption}[Positive smooth frontier density]\label{ass:frontier-density}
The density $h$ extends to a continuously differentiable function in neighborhoods of $-\Lb$ and $\Lb$, and
\[
  h(\Lb)>0.
\]
\end{assumption}

For this appendix, let $\widetilde C_T^\theta$ and $\widetilde C_T^A$ denote the high-effort wage bills obtained by inserting the minimum required bonuses into the two binary contracts, without imposing the cap $\wbar$. Whenever the corresponding bounded contract is feasible, $\widetilde C_T^\theta=C_T^\theta$ and $\widetilde C_T^A=C_T^A$.

\begin{proposition}[Frontier and attenuation rates]\label{prop:rate}
Suppose $\alpha_H<\rb$ and Assumption~\ref{ass:frontier-density} holds. Let
\[
  \delta_t:=\rb-\hat r_t.
\]
Then
\begin{equation}\label{eq:delta-rate}
  \delta_t
  =
  \frac{2\rb^2(1-\rb)}{h(\Lb)}\frac1t
  +O\!\left(\frac{\log t}{t^2}\right),
\end{equation}
and
\begin{equation}\label{eq:Pi-rate}
  \Pi_T=\Theta(T^{-2}).
\end{equation}
Consequently,
\begin{align}
  D(\theta,a_T)&=\Theta(T^{-2}),\label{eq:D-out-rate}\\
  \Del(2\rb-1)-D(a_0,a_T)&=O(T^{-1}),\label{eq:D-agree-rate}\\
  b_T^\theta,\ \widetilde C_T^\theta&=\Theta(T^2),\label{eq:out-contract-rate}\\
  b_T^A-b_\infty^A,
  \ \widetilde C_T^A-C_\infty^A&=O(T^{-1}).\label{eq:agree-contract-rate}
\end{align}
\end{proposition}

\begin{proof}
Set
\[
  B:=\frac{h(\Lb)}{\rb^2(1-\rb)}.
\]
For $r=\rb-x$ with $x\downarrow0$, write $s(r)=\Lb-u$. A Taylor expansion of the logit gives
\begin{equation}\label{eq:u-x}
  u=\frac{x}{\rb(1-\rb)}+O(x^2).
\end{equation}
Using $F'(r)=H(s(r))-H(-s(r))$,
\begin{align*}
  1-F'(\rb-x)
  &=[1-H(\Lb-u)]+H(-\Lb+u)\\
  &=[h(\Lb)+h(-\Lb)]u+O(u^2).
\end{align*}
The LLR identity gives $h(-\Lb)=e^{-\Lb}h(\Lb)$, and
\[
  1+e^{-\Lb}=\frac1\rb.
\]
Combining this identity with \eqref{eq:u-x} yields
\begin{equation}\label{eq:Fprime-boundary}
  1-F'(\rb-x)=Bx+O(x^2).
\end{equation}
Since $F(\rb)=\rb$,
\begin{align*}
  \delta_{t+1}
  &=\rb-F(\rb-\delta_t)\\
  &=\int_{\rb-\delta_t}^{\rb}F'(v)\,dv\\
  &=\delta_t-\frac{B}{2}\delta_t^2+O(\delta_t^3).
\end{align*}
For sufficiently large $t$, this recursion implies $\delta_t=\Theta(1/t)$. Moreover,
\[
  \frac1{\delta_{t+1}}-\frac1{\delta_t}
  =\frac{B}{2}+O(\delta_t).
\]
Summing and using $\sum_{k\le t}\delta_k=O(\log t)$ gives
\[
  \frac1{\delta_t}=\frac{B}{2}t+O(\log t),
\]
which is equivalent to \eqref{eq:delta-rate}.

Next, \eqref{eq:Fprime-boundary} and \eqref{eq:delta-rate} imply
\[
  \lambda_t
  =F'(\hat r_t)
  =1-B\delta_t+O(\delta_t^2)
  =1-\frac2t+O\!\left(\frac{\log t}{t^2}\right).
\]
Therefore
\[
  \log\Pi_T
  =\sum_{t=0}^{T-1}\log\lambda_t
  =-2\log T+O(1),
\]
which proves \eqref{eq:Pi-rate}.

Equation \eqref{eq:D-out-rate} follows from $D(\theta,a_T)=\Del\Pi_T$. For the process record,
\begin{align*}
  \Del(2\rb-1)-D(a_0,a_T)
  &=\Del\left[2(\rb-\hat r_T)+(2\alpha_H-1)\Pi_T\right],
\end{align*}
which is $O(T^{-1})$. The outcome bonus and uncapped wage-bill formulas then give \eqref{eq:out-contract-rate}. The process bonus and uncapped wage bill are smooth rational functions of $(\hat r_T,\Pi_T)$ in a neighborhood of $(\rb,0)$ with a strictly positive limiting denominator, so \eqref{eq:agree-contract-rate} follows from \eqref{eq:delta-rate} and \eqref{eq:Pi-rate}.
\end{proof}

\section{Proofs and derivations for Section~\ref{sec:beyond}}\label{app:beyond-proofs}

\begin{proof}[Proof of Proposition~\ref{prop:universal}]
The proof of Proposition~\ref{prop:endpoint-D} used only the initial-report decomposition and sign symmetry, not memory-1. It therefore gives \eqref{eq:univ-out} and \eqref{eq:univ-agree} for every maintained architecture.

Under \eqref{eq:positive-orientation}, compare the two positive branch combinations:
\[
  (S_T+R_T-1)-(S_T-R_T)=2R_T-1.
\]
Thus process evaluation has the larger distance if and only if $R_T\ge1/2$. The bonus and cost comparisons repeat the algebra in Theorem~\ref{thm:repair-threshold}; in particular, the cost difference remains
\[
  -\frac{\kappa S_T(2R_T-1)}
  {\Del(S_T-R_T)(S_T+R_T-1)}.
\]
\end{proof}

\begin{proof}[Proof of Corollary~\ref{cor:asym-learning}]
Proposition~\ref{prop:universal} and \eqref{eq:both-branches-learn} imply
\[
  D(\theta,a_T)=\Del|S_T-R_T|\to0
\]
and
\[
  D(a_0,a_T)=\Del|S_T+R_T-1|\to\Del.
\]
The pair $(a_0,a_T)$ is a measurable function of the full history, so Proposition~\ref{prop:data-processing} gives
\[
  D(H_T)\ge D(a_0,a_T).
\]
The general upper bound \eqref{eq:D-upper} gives $D(H_T)\le\Del$. The squeeze theorem yields $D(H_T)\to\Del$.
\end{proof}

\begin{proof}[Derivation for Example~\ref{ex:cascade}]
At a neutral public history, two successive informative actions both favor the true state with probability $q^2$, both favor the false state with probability $(1-q)^2$, and disagree with probability $2q(1-q)$. Disagreement restores neutrality. This gives \eqref{eq:Gamma-recursion} and \eqref{eq:Gamma}.

If the initial report is correct, the first downstream private signal agrees with it with probability $q$, creating a correct two-action cascade. With probability $1-q$, the two actions disagree and the process returns to neutrality. Hence
\[
  S_\infty=q+(1-q)\Gamma.
\]
If the initial report is incorrect, a private signal that also points in the wrong direction, probability $1-q$, creates an incorrect cascade. A correct private signal, probability $q$, restores neutrality. Hence
\[
  R_\infty=q\Gamma.
\]
Using \eqref{eq:Gamma},
\begin{align*}
  S_\infty-R_\infty
  &=q+(1-2q)\Gamma
   =\frac{q(1-q)}{q^2+(1-q)^2},\\
  S_\infty+R_\infty-1
  &=q-1+\Gamma\\
  &=\frac{(2q-1)(q^2-q+1)}{q^2+(1-q)^2}.
\end{align*}
Multiplication by $\Del$ gives \eqref{eq:cascade-out} and \eqref{eq:cascade-agree}.
\end{proof}

\begin{proof}[Proof of Corollary~\ref{cor:cascade-threshold}]
The universal repair threshold requires
\[
  R_\infty
  =\frac{q^3}{q^2+(1-q)^2}
  \ge\frac12.
\]
This is equivalent to
\[
  g(q):=2q^3-2q^2+2q-1\ge0.
\]
Now
\[
  g'(q)=6q^2-4q+2>0
\]
for every $q$, because its discriminant is negative. Moreover,
\[
  g(1/2)=-1/4<0,
  \qquad
  g(1)=1>0.
\]
There is therefore a unique root $q^\star\in(1/2,1)$, and the ranking follows from Proposition~\ref{prop:universal}. Numerical root finding gives $q^\star\approx0.647799$.
\end{proof}

\section{Proofs for Section~\ref{sec:asymmetry}}\label{app:asymmetry-proofs}

\begin{proof}[Proof of Proposition~\ref{prop:asym-regimes}]
The two frontiers divide the real line into three regions: below $-\Lb$, inside $(-\Lb,\Lb)$, and above $\Lb$. Because $u_0>v_0$, the ordered pair $(u_0,v_0)$ can occupy only six of the nine cells in the resulting three-by-three grid. These six cells are exactly the six cases listed in the proposition.

If a post-report belief is above $\Lb$, then even the lowest possible private LLR cannot make the posterior log-odds negative, so the next reviewer chooses $+1$. If it is below $-\Lb$, even the highest possible private LLR cannot make the posterior log-odds positive, so the next reviewer chooses $-1$. If it lies strictly between the frontiers, some private-signal realizations overturn the inherited report and some do not. These observations give the economic descriptions of the six regimes. Boundary cases have probability zero under the no-atoms assumption and can be assigned to the adjacent uncontestable regimes without affecting the induced laws.
\end{proof}

\section{Numerical implementation details}\label{app:numerical}

For the density in \eqref{eq:numerical-density}, normalization follows from
\[
  \int_{-\Lb}^0 ce^z\,dz+\int_0^{\Lb}c\,dz
  =c(1-e^{-\Lb}+\Lb)=1.
\]
The CDF in \eqref{eq:numerical-cdf} follows by direct integration. The recursion used for the figures is
\begin{align*}
  s_t&=\log\frac{\hat r_t}{1-\hat r_t},\\
  \lambda_t&=H(s_t)-H(-s_t),\\
  \hat r_{t+1}
  &=\hat r_t[1-H(-s_t)]
    +(1-\hat r_t)[1-H(s_t)],\\
  \Pi_{t+1}&=\Pi_t\lambda_t,\\
  S_t&=\hat r_t+(1-\alpha_H)\Pi_t,\\
  R_t&=\hat r_t-\alpha_H\Pi_t.
\end{align*}
The source archive contains the script \texttt{generate\_figures.py} and the resulting comma-separated path. No simulation is needed: all plotted values are deterministic iterations of these equations.

\end{appendices}

\bibliographystyle{plainnat}
\bibliography{washed_out_refs}
\end{document}